\documentclass[aps,prl,twocolumn,superscriptaddress,10pt,article,showpacs,longbibliography]{revtex4-2}

\usepackage{blindtext}
\usepackage{lipsum}
\usepackage{graphics}
\usepackage{amsmath}
\usepackage{graphicx}
\usepackage{amssymb}
\usepackage{pifont}
\usepackage{verbatim}
\usepackage{physics}
\usepackage[normalem]{ulem}
\usepackage[dvipsnames]{xcolor}
\usepackage{bbm}
\usepackage{enumitem}
\usepackage{multirow}

\usepackage{bm}
\usepackage{physics}
\usepackage{amsfonts}
\usepackage{amsthm}
\usepackage{bbm}
\usepackage{mathtools}
\usepackage{braket}
\usepackage[normalem]{ulem}
\usepackage{wrapfig}
\usepackage{tikz}
\usepackage{dsfont}
\usepackage{comment}
\usepackage{thmtools,thm-restate}
\usepackage[export]{adjustbox}

\makeatletter
\AtBeginDocument{%
    \newwrite\bibnotes
    \def\bibnotesext{Notes.bib}
    \immediate\openout\bibnotes=\jobname\bibnotesext
    \immediate\write\bibnotes{@CONTROL{REVTEX41Control}}
    \immediate\write\bibnotes{@CONTROL{%
    apsrev42Control,author="08",editor="1",pages="1",title="0",year="1"}}
     \if@filesw
     \immediate\write\@auxout{\string\citation{apsrev41Control}}%
    \fi
}%
\makeatother

\def\C{\mathbb C}
\def\Z{\mathbb Z}

\def\rank{\mathsf{rank}}
\def\poly{\mathsf{poly}}

\usepackage{listings}
\usepackage{xcolor}
\usepackage{colortbl}
\usepackage{diagbox}
\usepackage{qcircuit}

\usepackage{algpseudocode}

\newcounter{protocol}
\renewcommand{\theprotocol}{\arabic{protocol}}
{%
  \refstepcounter{protocol}%
  \par\medskip\noindent
  \textbf{Protocol~\theprotocol%
  \ifx\relax#1\relax\else\ (\textit{#1})\fi.}%
  \par\smallskip\noindent\rule{\columnwidth}{0.4pt}\par\smallskip
}%
{%
  \smallskip\noindent\rule{\columnwidth}{0.4pt}\par\medskip
}

\newtheorem{theorem}{Theorem}
\newtheorem{definition}{Definition}

\newtheorem{lemma}{Lemma}
\newtheorem{corollary}{Corollary}
\newtheorem{proposition}{Proposition}

\newtheorem{conjecture}{Conjecture}

\newcommand{\sgn}{\operatorname{sgn}}
\newcommand{\id}{\mathbbm 1}
\newcommand{\M}{\mathrm{M}}

\definecolor{myrefcolor}{rgb}{0.067,0.5,0.5}
\usepackage[
    breaklinks,
    pdftex,
    colorlinks=true,
    linkcolor=myrefcolor,
    citecolor=myrefcolor,
    urlcolor=myrefcolor
]{hyperref}

\usepackage{times}
 
\begin{document}

\title{Measurement-induced non-commutativity in adaptive fermionic linear optics}

\author{Chenfeng Cao}
\affiliation{Dahlem Center for Complex Quantum Systems, Freie Universit\"{a}t Berlin, 14195 Berlin, Germany}
\affiliation{HK Institute of Quantum Science $\&$ Technology, The University of Hong Kong, Hong Kong, China}

\author{Yifan Tang}
\affiliation{Dahlem Center for Complex Quantum Systems, Freie Universit\"{a}t Berlin, 14195 Berlin, Germany}

\author{Jens Eisert}
\email{jense@zedat.fu-berlin.de}
\affiliation{Dahlem Center for Complex Quantum Systems, Freie Universit\"{a}t Berlin, 14195 Berlin, Germany}
\affiliation{Helmholtz-Zentrum Berlin f{\"u}r Materialien und Energie, 14109 Berlin, Germany}

\date{\today}

\begin{abstract}
Fermionic linear optics (FLO) with Gaussian resources is
known to be efficiently classically simulable. We show that this changes for fermions with internal degrees of freedom once the circuit is equipped with mid-circuit number monitoring and classical feedforward. In our architecture, the measurement record routes the selected blocks into a fixed-order Bell-fusion pairing geometry. On the level of the classical description, the permutation sum then no longer collapses to a single determinant or Pfaffian. Each accepted branch expands as a signed sum of path-ordered products of typically non-commuting dressed blocks, and branch amplitudes are matrix elements of the resulting non-commutative trace polynomials. Numerically, we observe Porter--Thomas statistics for the conditional branch distribution and rapid growth of the minimal order-respecting matrix product operator bond dimension. Under reasonable complexity assumptions, these results identify mid-circuit measurement-induced non-commutativity as a route to sampling hardness for noninteracting fermions, with the non-Gaussianity confined to a fixed preparation-and-readout interface rather than to the coherent FLO evolution.
\end{abstract}

\maketitle

Quantum random sampling protocols
aimed at showing a quantum advantage target restricted dynamics whose output distributions are believed to be hard to sample classically, even without full fault tolerance~\cite{Harrow2017Quantum,Hangleiter2023Computational,Eisert2025Mind}. These proposals have motivated experiments across platforms, from superconducting random-circuit sampling to large-scale photonic sampling~\cite{Arute2019Quantum,Zhong2020Quantum,Madsen2022Quantum}. Although such proposals have a complexity-theoretic underpinning, the hardness evidence remains indirect in practice: asymptotic arguments must be weighed against finite-depth noise and steadily improving classical simulation and spoofing methods~\cite{Boixo2018Characterizing,Bouland2019On,Tindall2024Efficient,Gao2024Limitations,Schuster2025A}.

A central feature of these proposals is that output amplitudes encode algebraic quantities that are $\#\mathrm{P}$-hard to compute. 
In boson sampling, amplitudes are matrix permanents~\cite{Valiant1979The,Aaronson2011The,Brod2019Photonics}, 
while Gaussian variants access hafnians and related matrix functions~\cite{Lund2014Boson,Hamilton2017Gaussian}. 
In commuting (IQP) circuit models, amplitudes map to Ising partition functions and related invariants that are hard to compute and conjecturally hard on average~\cite{Bremner2011Classical,Bremner2016Average-Case,Bremner2017Achieving,NewSupremacy}. Using Stockmeyer's approximate-counting machinery~\cite{Hangleiter2023Computational}, one can lift the hardness of computing output probabilities to the hardness of approximate sampling in total variation distance. Mid-circuit measurements with adaptive feedforward have also emerged as a useful architectural resource in this context~\cite{Paletta2024Robust,Cao2026Measurement}; in bosonic Gaussian circuits, recent work identifies the number of adaptive measurement-and-feedforward steps as a natural complexity parameter for mean-value estimation~\cite{Oh2026Classical}. 

By contrast, \emph{fermionic linear optics} (FLO) with Gaussian resources is efficiently classically simulable, with many-body amplitudes reducing to determinants and Pfaffians, and matchgate circuits admitting polynomial-time simulation~\cite{Knill2001Fermionic,Terhal2002Classical,Valiant2002Quantum,Bravyi2005Lagrangian,Jozsa2008Matchgates,ReardonSmith2024Improved}. Consequently, proposals for fermionic quantum advantages, 
often referred to as ``fermion sampling'', rely on injecting 
non-Gaussian resources such as interaction gadgets or 
magic-state inputs~\cite{Hebenstreit2019All, Hebenstreit2020Computational, Oszmaniec2022Fermion}. Adaptive non-Gaussian charge/parity measurements are also known to restore universal quantum computation for fermionic linear optics via gate-synthesis constructions~\cite{Beenakker2004Charge,DiVincenzo2005Fermionic}; however, these results do not characterize the algebraic structure of monitored branch amplitudes, nor do they address the sampling-hardness question considered here.

\begin{figure*}[t]
\centering
\includegraphics[width=\textwidth]{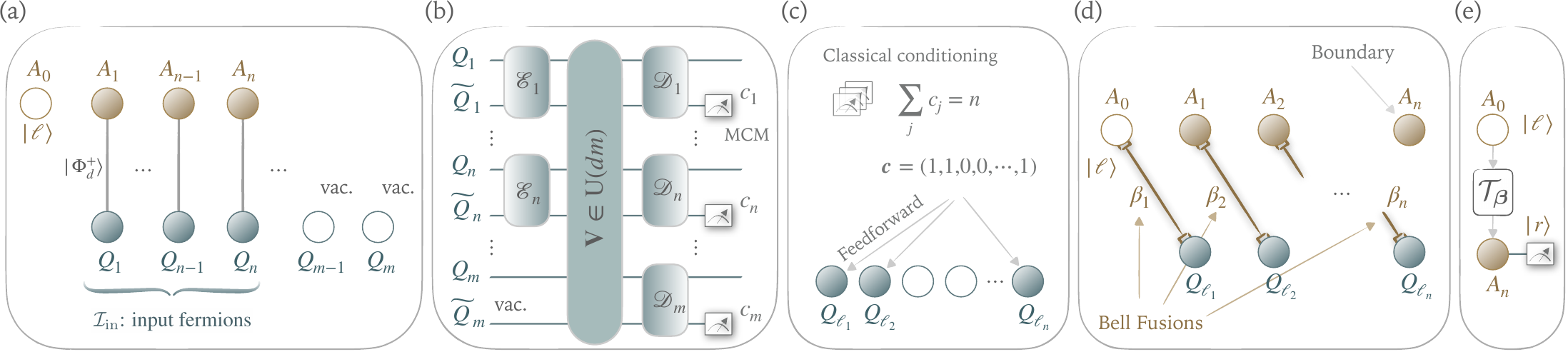}
\caption{%
Monitored adaptive FLO protocol. Panels (a)--(e) realize protocol steps (i)--(v).
(a)~Bell initialization: each input block $Q_j$, $j\in\mathcal I_{\rm in}=\{1,\dots,n\}$, is prepared in a generalized Bell pair $|\Phi_d^+\rangle$ with the auxiliary qudit $A_j$. The boundary auxiliary $A_0$ starts in $|\ell\rangle$ and all other blocks in $|{\rm vac}\rangle$.
(b)~Encoding, FLO evolution, and mid-circuit monitoring (MCM): the encoders $\mathcal E_j$ (input blocks only) write the internal label onto the occupied orbital, the number-conserving propagator $\mathbf V\in\mathrm U(dm)$ acts on all encoded orbitals, and the decoders $\mathcal D_j$ transfer occupations onto flag qubits whose readout yields the record $\boldsymbol c$ without resolving internal labels.
(c)~Heralded acceptance and feedforward: a collision-free record ($\sum_j c_j=n$) selects blocks $\ell_1<\cdots<\ell_n$, fixes $S_{t,k}=V_{\ell_t,k}$, and is fed forward to route them into the fusion geometry.
(d)~Ordered Bell fusion: step $t$ fuses the pair $(A_{t-1},Q_{\ell_t})$ with outcome $\beta_t$, in the fixed order $t=1,\dots,n$.
(e)~Boundary readout: measuring $A_n$ (outcome $r$) completes the accepted branch, whose amplitude is $\mathcal A_d(\boldsymbol c,\boldsymbol\beta)=\langle r|\mathcal T_{\boldsymbol\beta}(\mathbf S)|\ell\rangle$ with fixed-boundary branch law \(p(\boldsymbol\beta\,|\,\boldsymbol c)\propto|\mathcal A_d|^2\), in the convention of step~(v). The byproducts from $\boldsymbol\beta$ enter the branch operator, with no coherent correction applied.%
}
\label{fig:architecture}
\end{figure*}

In this work, we consider a monitored free-fermion architecture in which fermions carrying an internal $d$-level label first propagate through a number-conserving FLO circuit and are then subjected to a coarse-grained blockwise occupation measurement. This monitoring asks only which spatial blocks are singly occupied, without resolving the internal orbital, and the resulting collision-free record is used to route the selected blocks into a fixed Bell-fusion readout geometry. This routing is the only adaptive operation of the protocol: the fusions are fixed in advance and their outcomes require no correction, so a single round of feedforward suffices. Because the record specifies which block enters which fusion step, the antisymmetrized FLO interference becomes an order-sensitive contraction problem. This coarse-grained projection onto the single-occupation sector is non-Gaussian and lies outside matchgate/FLO simulation frameworks based on local Gaussian, mode-resolved measurements~\cite{Terhal2002Classical,Jozsa2008Matchgates,Brod2016Efficient}. Unlike the magic-state and interaction-gadget constructions~\cite{Hebenstreit2019All,Hebenstreit2020Computational,Oszmaniec2022Fermion}, no dedicated fermionic magic state is injected into the propagating register and no interaction gadget is synthesized. The propagating fermionic register undergoes strictly number-conserving FLO. Unlike the gate-synthesis route~\cite{Beenakker2004Charge,DiVincenzo2005Fermionic}, the measurement record is retained as branch data rather than corrected into a deterministic gate.

We formalize this picture through an explicit monitored-FLO protocol on $m$ spatial blocks with $n$ input fermions (Fig.~\ref{fig:architecture}). Conditioned on a collision-free monitoring record, obtained at constant average rate in the dilute regime $m=\Theta(n^2)$ (Lemma~\ref{lem:post-selection-rate}), the encoded FLO stage induces a determinantal operator kernel on the selected blocks. The subsequent Bell-fusion outcomes pin teleportation byproducts to fixed fusion steps~\cite{Bennett1993Teleporting,Gottesman1999Demonstrating,Stephen2024Preparing,Smith2023Deterministic,Bartolucci2023Fusion}, so that each accepted branch is governed by an outcome-dependent operator on the boundary wire whose matrix elements are the branch amplitudes. Its algebraic form is that of a matrix-valued non-commutative polynomial in byproduct-dressed blocks. Measurement and feedforward thus map free-fermion interference onto non-commutative matrix multiplication.

To quantify classical simulability, we consider contractions that respect the enforced contraction order. Each branch operator admits an exact \emph{matrix product operator} (MPO) representation in this order, with one local tensor per fusion step, whose minimal bond dimension lower-bounds the sequential memory of any single-pass contraction that follows the enforced order. Numerically, typical branches exhibit rapid growth of this minimal bond dimension together with Porter--Thomas statistics for the conditional branch weights. Motivated by Haar-induced random-FLO ensembles, we formulate an average-case hardness conjecture for conditional branch probabilities, in direct analogy with conjectures underlying other sampling-advantage proposals~\cite{Aaronson2011The,Bremner2016Average-Case}. With accepted-ensemble anti-concentration and Stockmeyer approximate counting, it would imply classical intractability of approximate sampling.

\smallskip
\paragraph*{Architecture and protocol.}
Each spatial block of the quantum architecture contains $d$ fermionic orbitals; the occupied orbital carries an internal $d$-level label. To implement the monitoring and fusion readout, we consider two registers per block $j\in\{1,\dots,m\}$, a qudit $Q_j\simeq\C^d$ and an auxiliary register $\widetilde Q_j$ of $(d-1)$ qubits, together with an auxiliary qudit chain $(A_0,A_1,\dots,A_n)$ with $A_j\simeq\C^d$ [Fig.~\ref{fig:architecture}(a)]. The pair $(Q_j,\widetilde Q_j)$ encodes the local vacuum and single-particle subspace of $d$ fermionic orbitals $\{a^\dagger_{j,\alpha}\}_{\alpha=0}^{d-1}$, with $a^\dagger_{j,\alpha}|{\rm vac}\rangle_j$ carrying the internal label $\alpha$; configurations with local multi-occupancy fall outside the code space and are rejected at the monitoring step. A number-conserving FLO circuit acts on the encoded $dm$ orbitals and induces a single-particle propagator $\mathbf V\in \mathrm U(dm)$, viewed as an $m\times m$ block matrix $\mathbf V=(V_{j,k})$ with $V_{j,k}\in\M_d(\C)$.

On each block $j$ we fix local maps on $(Q_j,\widetilde Q_j)$ implementing an encoding $\mathcal E_j$ and decoding $\mathcal D_j$; for fixed $d$ they incur only constant local overhead, and an explicit qubit construction for $d=2$ is given in the End Matter. The encoding identifies the internal label $\alpha$ with the occupied orbital,
\begin{equation}
\label{eq:local-encode}
\mathcal E_j:\ |\alpha\rangle_{Q_j}|0^{d-1}\rangle_{\widetilde Q_j}\mapsto
a^\dagger_{j,\alpha}|{\rm vac}\rangle_j,\quad \alpha=0,\dots,d-1.
\end{equation}
The decoding $\mathcal D_j$ inverts this map while writing an occupation flag $c_j\in\{0,1\}$ onto a designated qubit $f_j\subset\widetilde Q_j$, so that measuring $f_j$ reveals only whether the block lies in the single-particle sector and does not resolve $\alpha$, while preserving the coherence of the internal state on $Q_j$ up to a fixed known local basis unitary.

As in
Fig.~\ref{fig:architecture}(a) we occupy the first $n$ sites
$\mathcal I_{\rm in}=\{1,\dots,n\}$.
\begin{enumerate}[label=(\roman*),leftmargin=*,itemsep=1pt,topsep=2pt]
\item \emph{Bell initialization} [Fig.~\ref{fig:architecture}(a)].
For each $j\in\mathcal I_{\rm in}$ prepare a generalized Bell pair
$|\Phi_d^+\rangle=d^{-1/2}\sum_{\alpha=0}^{d-1}|\alpha,\alpha\rangle$ on $(Q_j,A_j)$.
Initialize the boundary auxiliary $A_0$ in $|\ell\rangle$ and all unoccupied blocks in
$|{\rm vac}\rangle$.

\item \emph{Encoding and FLO evolution} [Fig.~\ref{fig:architecture}(b)].
Apply $\mathcal E_j$ on $j\in\mathcal I_{\rm in}$ and run the number-conserving FLO circuit with propagator $\mathbf V$.
During FLO evolution the internal label is the occupied orbital index. \(\mathbf V\) may coherently mix spatial and internal modes through \(V_{j,k}\).

\item \emph{Mid-circuit monitoring (MCM) and heralded acceptance} [Fig.~\ref{fig:architecture}(b,\,c)].
After decoding, measure the designated flag qubit $f_j\subset\widetilde Q_j$ and define $c_j=1$ if and only if block $j$ is in the decoded single-particle sector. The record $\boldsymbol c$ is produced by the Born rule; we accept the run when $\boldsymbol c$ is \emph{collision-free} with exactly $n$ selected blocks $\ell_1<\cdots<\ell_n$, and the realized accepted record is retained as public classical side information that fixes the feedforward routing.
This fixes the selected block sub-matrix $\mathbf S=(S_{t,k})_{t,k=1}^n$ with $S_{t,k}:=V_{\ell_t,k}$.

\item \emph{Feedforward routing and Bell fusion} [Fig.~\ref{fig:architecture}(c,\,d)].
Use $\boldsymbol c$ to classically determine a routing that maps the selected blocks into a fixed Bell-fusion pairing geometry, and perform the nearest-neighbor generalized Bell-fusion readout on the pairs $(A_{t-1},Q_{\ell_t})$, $t=1,\dots,n$, with outcomes $\boldsymbol\beta=(\beta_1,\dots,\beta_n)$ and boundary systems $A_0$ and $A_n$.

\item \emph{Boundary readout} [Fig.~\ref{fig:architecture}(e)].
Measure the boundary auxiliary $A_n$ in a fixed basis, obtaining $r\in\{0,\dots,d-1\}$. Operationally $r$ is part of the output, giving the formal law $p(\boldsymbol\beta,r\,|\,\boldsymbol c)$. For the algebra and numerics we fix $|\ell\rangle=|r\rangle=|0\rangle$, selecting one matrix element of the resulting branch operator.

\end{enumerate}

\medskip
\noindent
In the dilute regime $m=\Theta(n^2)$ the collision-free heralding probability is bounded away from zero.
\begin{lemma}[Constant-rate collision-free monitoring]
\label{lem:post-selection-rate}
Fix $d\ge2$ and $\kappa>0$, and let $p_{\rm e}(n,m,d)$ denote the probability, jointly
over $\mathbf V$ and the Born rule in step~(iii), that monitoring yields a collision-free
record with exactly $n$ effective blocks.
If \(m=\lceil \kappa n^2\rceil\), then there exists $c_{\kappa,d}>0$ independent of $n$ such that
$p_{\rm e}(n,m,d)\ge c_{\kappa,d}$ for all sufficiently large $n$.
\end{lemma}
\noindent
A proof is given in the Supplemental Material (SM). Figure~\ref{fig:collision} provides a numerical check against the asymptotic constant $p_\infty=\exp[-(1-1/d)/(2\kappa)]$. Unless stated otherwise, all subsequent numerical data use $\kappa=1/2$ with $m=\lceil n^2/2\rceil$. Moreover, a constant fraction of Haar instances have a constant per-instance heralding probability (SM). Below, the accepted ensemble means that $\mathbf V$ is drawn first and $\boldsymbol c$ is then Born-sampled conditioned on acceptance.

\medskip
\noindent
To express the fermionic exchange explicitly, we use the permutation operators $\mathcal P_\sigma$ on $(\C^d)^{\otimes n}$.

\begin{lemma}[FLO-induced determinantal kernel]
\label{lem:ftpd}
Conditioned on a realized collision-free monitoring record $\boldsymbol c$, which selects output blocks $\ell_1<\cdots<\ell_n$ and hence the block sub-matrix $\mathbf S$, the encoded FLO stage induces on the selected $n$ blocks the operator
\begin{equation}
\label{eq:ftpd-maintext}
\det_{\otimes}(\mathbf S)
:= \sum_{\sigma\in \mathfrak{S}_n}\sgn(\sigma)
 \Bigl(\bigotimes_{t=1}^n S_{t,\sigma(t)}\Bigr)\mathcal P_{\sigma^{-1}}.
\end{equation}
\end{lemma}

\noindent
Here $\mathfrak S_n$ denotes the symmetric group, see SM for the proof.
Lemma~\ref{lem:ftpd} gives the accepted-branch FLO kernel to be contracted by the monitored fusion layer.

\smallskip
\paragraph*{Branch amplitudes and non-commutative structure.}
Let $X,Z$ denote the generalized Pauli operators on $\C^d$. For a single Bell outcome
$\beta=(a,b)\in\Z_d^2$, define the Weyl operator
$\hat D_{\beta}:=X^a Z^b$.
We use the Bell basis $|\Phi_{\beta}\rangle := (\id\otimes \hat D_{\beta}^\dagger)|\Phi_d^+\rangle$.
At fusion step $t$, the selected block $S_{t,k}$ entering that step and the Bell-measurement outcome $\beta_t$ insert the unnormalized matrix $d^{-1}S_{t,k}^{\top}\hat D_{\beta_t}$ into the ordered boundary-wire product, where ${}^{\top}$ is the computational-basis transpose. This update follows from a three-party Bell-projection identity derived in the SM, where the transpose-induced relabeling of $\beta$ (for $d>2$) is absorbed into the outcome label.

\begin{figure}[t]
\centering
\includegraphics[width=1\columnwidth]{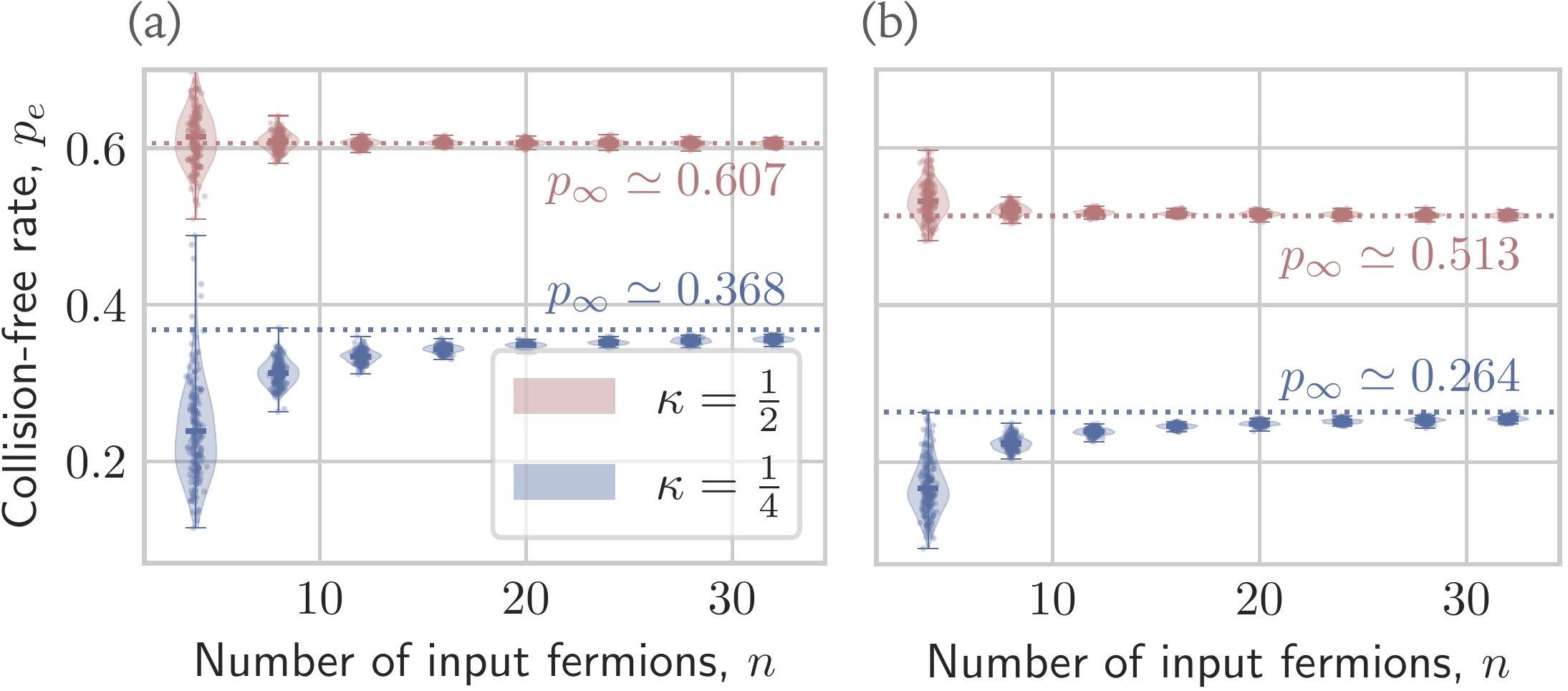}
\caption{
Collision-free heralding probability $p_{\rm e}(n,m,d)$ in the dilute regime (\(m=\lceil\kappa n^2\rceil\)).
Empirical distributions of per-instance rates (over 200 Haar-random realizations) are compared against the asymptotic prediction $p_\infty=\exp[-(1-1/d)/(2\kappa)]$ (dotted lines, Lemma~\ref{lem:post-selection-rate}).
(a) $d=2$ and (b) $d=3$, both with $\kappa\in\{1/2,1/4\}$.
}
\label{fig:collision}
\end{figure}

Fix fusion outcomes ${\boldsymbol\beta}=(\beta_1,\dots,\beta_n)\in(\Z_d^2)^n$, where each
$\beta_t=(a_t,b_t)$, and label fusion steps in the fixed order
$t=1,\dots,n$ (after feedforward routing).
We absorb byproducts into \emph{dressed blocks} $\widetilde S_{t,k}:=S_{t,k}^{\top}\hat D_{\beta_t}$.
The fusion layer defines a linear contraction map $\mathcal C_{\boldsymbol\beta}$ that inserts the
Bell-projection updates at the prescribed fusion steps, yielding the branch operator
\begin{equation}
\label{eq:T-def}
\mathcal{T}_{\boldsymbol\beta}(\mathbf S)\ :=\ \mathcal C_{\boldsymbol\beta}\left(\det_{\otimes}(\mathbf S)\right)\ \in\ \M_d(\C).
\end{equation}
(Here $\mathcal C_{\boldsymbol\beta}$ includes the fixed $d^{-1}$ factors from each Bell-fusion readout branch.)
For this fixed accepted record, the branch amplitude is $\mathcal A_d(\boldsymbol c,{\boldsymbol\beta})=\langle r|\mathcal{T}_{\boldsymbol\beta}(\mathbf S)|\ell\rangle$. Throughout, $p(\boldsymbol\beta\,|\,\boldsymbol c)\propto|\mathcal A_d(\boldsymbol c,{\boldsymbol\beta})|^2$ denotes this fixed-boundary branch law, the convention used in the figures and formal statements below. The SM formulates the task also over the full output $y=(\boldsymbol\beta,r)$. The hardness argument there applies in either convention, with the average-case hardness conjecture and the anti-concentration assumption read in that same convention. For fixed $d$, changing $|\ell\rangle,|r\rangle$ only selects a matrix element of the same branch operator $\mathcal{T}_{\boldsymbol\beta}(\mathbf S)$ and does not affect its ordered polynomial structure.

\begin{proposition}[Non-commutative branch structure]
\label{prop:nc-amp}
Fix fusion outcomes ${\boldsymbol\beta}=(\beta_1,\dots,\beta_n)$ and define the dressed blocks
$\widetilde S_{t,k}:=S_{t,k}^{\top}\hat D_{\beta_t}$ in the fixed contraction order $t=1,\dots,n$.
Then $\mathcal{T}_{\boldsymbol\beta}(\mathbf S)$ is an \emph{ordered non-commutative trace polynomial} in the dressed blocks, with scalar trace factors generated by closed wiring loops.
Formally, $\mathcal{T}_{\boldsymbol\beta}(\mathbf S)$ is the matrix evaluation of an element of the free trace-polynomial algebra \(\C\langle \widetilde S\rangle_{\Tr}\), the free algebra in the non-commuting generators $\widetilde S_{t,k}$ extended by central trace symbols for cyclic words.
\end{proposition}

\medskip
\noindent
The key mechanism is a path--cycle decomposition induced by the fixed fusion order:
each permutation wiring decomposes into a unique open boundary path and disjoint directed cycles.
The path contributes a left-multiplied ordered product of dressed blocks, while each cycle yields a scalar trace-loop factor.
Already for $n=2$, $\mathcal{T}_{\boldsymbol\beta}(\mathbf S)
=\frac{1}{d^2}(\widetilde S_{2,2}\widetilde S_{1,1}
-{\Tr}(\widetilde S_{2,1})\,\widetilde S_{1,2} )$,
where the swap contributes the loop factor ${\Tr}(\widetilde S_{2,1})$.
See End Matter for the general derivation and an explicit $n=3$ expansion, and the SM for the formal construction and proof. In particular, each permutation term contributes exactly one dressed block from each fusion step $t$, so
$\mathcal{T}_{\boldsymbol\beta}(\mathbf S)$ is multilinear in the row-$t$ blocks $\{\widetilde S_{t,k}\}_{k=1}^n$.

\begin{figure}[t]
\centering
\includegraphics[width=1\columnwidth]{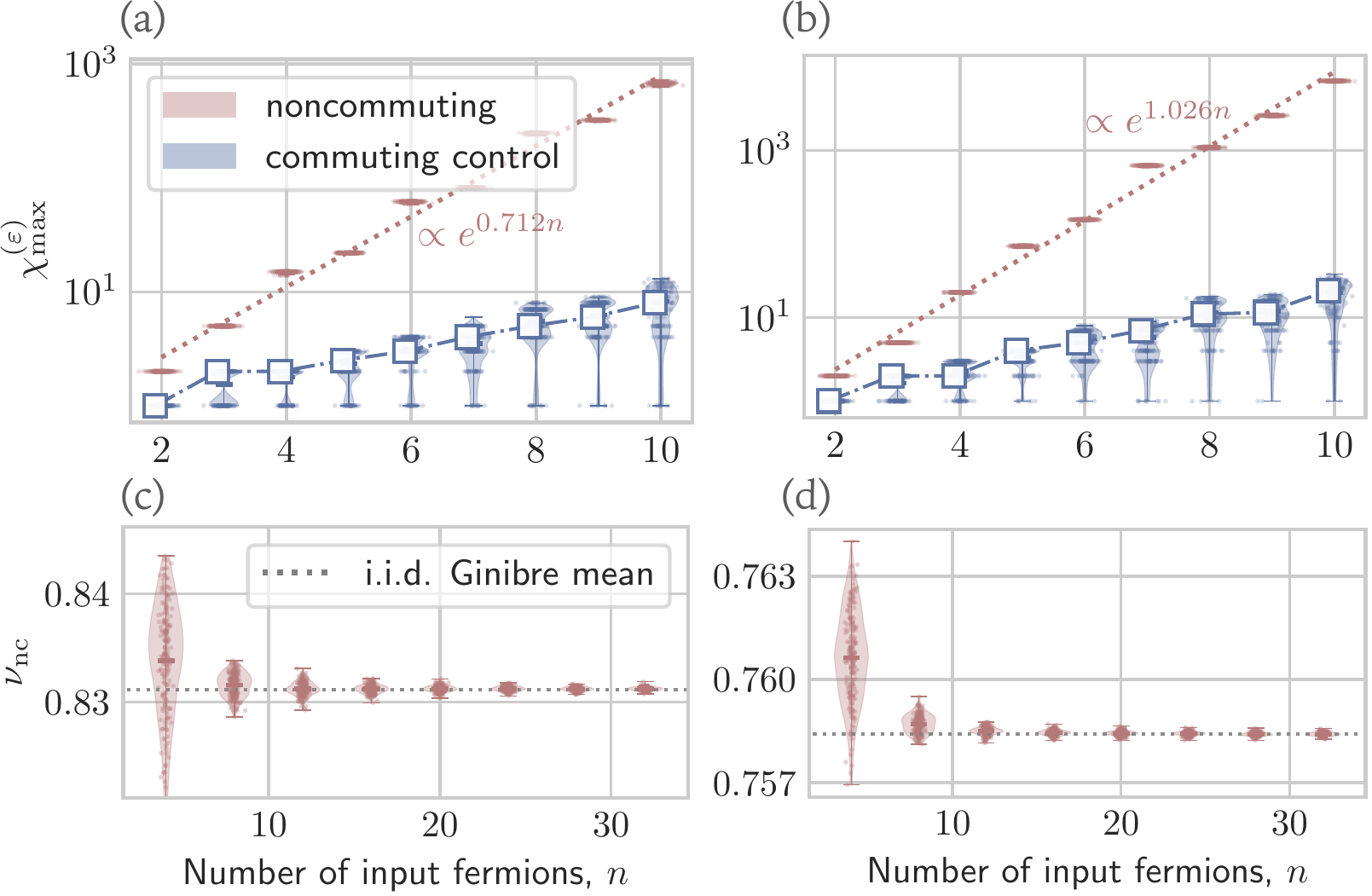}
\caption{
Order-respecting contraction memory and non-commutativity diagnostics for monitored-FLO branches.
(a, b) Truncated minimal MPO bond dimension $\chi_{\max}^{(\varepsilon)}$ ($\varepsilon=10^{-3}$) in the enforced fusion order
(red: monitored-FLO; blue: commuting control with $S_{t,k}\propto \id_d$ and $Z$-only byproducts; dotted: exponential fit).
(c, d) Normalized commutator diagnostic $\nu_{\mathrm{nc}}$ of the dressed blocks [Eq.~\eqref{eq:nc-score-def}];
dotted line: i.i.d.\ Ginibre reference mean. Panels (c, d) are shown for larger $n$ to highlight the concentration of $\nu_{\mathrm{nc}}$ with system size.
(a, c) $d=2$. (b, d) $d=3$. Each instance is a Haar $\mathbf V$ with a collision-free record $\boldsymbol c$ Born-sampled from the accepted ensemble (200 per $n$).
}
\label{fig:bond}
\end{figure}

\smallskip
\paragraph*{Numerical diagnostics.}
In the enforced fusion order each branch reduces to a one-dimensional sweep over steps $t=1,\ldots,n$ with local updates linear in the dressed blocks. Concretely, $\mathcal{T}_{\boldsymbol\beta}(\mathbf S)$ admits an exact fusion-order--respecting operator-valued MPO representation
\begin{equation}
\label{eq:mpo-rep}
\mathcal{T}_{\boldsymbol\beta}(\mathbf S)
=\sum_{x_1,\ldots,x_{n-1}}
\mathsf{T}^{(n)}_{{\boldsymbol\beta},(x_{n-1},x_n)}\cdots
\mathsf{T}^{(2)}_{{\boldsymbol\beta},(x_1,x_2)}\,
\mathsf{T}^{(1)}_{{\boldsymbol\beta},(x_0,x_1)}.
\end{equation}
Here $x_0$ and $x_n$ are fixed boundary bond indices (equivalently absorbed into boundary vectors), and the sum runs over $x_1,\ldots,x_{n-1}$.
Let $\chi_t$ denote the minimal bond dimension across cut $t$ among all such fusion-order--respecting representations, and set $\chi_{\max}:=\max_t\chi_t$.
Then $\chi_{\max}$ lower-bounds the virtual memory required by any single-pass contraction that processes the fusion steps in the enforced order~\cite{Schollwock2011The,Orus2014A}.
In our numerics we report the $\varepsilon$-truncated value $\chi_{\max}^{(\varepsilon)}$ (with $\varepsilon=10^{-3}$), computed by sequential partial contractions in the enforced fusion order with singular-value truncation at tolerance $\varepsilon$.
Figs.~\ref{fig:bond}(a, b) show that $\chi_{\max}^{(\varepsilon)}$ grows rapidly with $n$ in the monitored-FLO ensemble, while a commuting-control benchmark strongly suppresses this growth.

In a symbolic benchmark where the dressed blocks are treated as algebraically independent non-commuting generators (and trace-loop factors as formal scalars), this sequential memory is provably exponential:

\begin{theorem}[Exponential sequential-memory barrier]
\label{thm:memory-barrier}
In the above generic model, the minimal fusion-order--respecting bond dimension obeys
$\chi_t=\Omega_d \bigl(\binom{n}{t}\bigr)$ for $t=1,\ldots,n-1$, and in particular
$\chi_{\lfloor n/2\rfloor} =2^{\Omega(n)}$.
\end{theorem}
\noindent
Intuitively, across cut $t$ a sequential contraction must distinguish which subset of the $n$ occupied input legs has been
routed through the first $t$ fusion steps; in the non-commutative regime these sectors are generically linearly independent.
A proof is given in the SM.

\smallskip
\noindent
Although Theorem~\ref{thm:memory-barrier} is proven in the free trace-polynomial benchmark, where the dressed blocks are treated as algebraically independent non-commuting generators,
in the monitored-FLO ensemble the $\widetilde S_{t,k}=S_{t,k}^{\top}\hat D_{\beta_t}$ are correlated subblocks of a single Haar unitary.
To assess whether these correlations suppress non-commutativity in the dilute regime, we report the normalized commutator diagnostic
\begin{equation}
\label{eq:nc-score-def}
\nu_{\mathrm{nc}}(\boldsymbol c,{\boldsymbol\beta})
:= \mathbb E_{(t,k)\neq(t',k')}
\frac{\bigl\|[\widetilde S_{t,k},\widetilde S_{t',k'}]\bigr\|_{F}}
{\|\widetilde S_{t,k}\|_{F}\,\|\widetilde S_{t',k'}\|_{F}},
\end{equation}
where $\|\cdot\|_F$ denotes the Frobenius norm and the average is taken over distinct pairs of dressed blocks.
Figs.~\ref{fig:bond}(c, d) show that $\nu_{\mathrm{nc}}$ remains close to the corresponding i.i.d.\ Ginibre reference mean~\cite{Ginibre1965Statistical},
indicating that typical dilute instances remain strongly non-commutative rather than effectively commuting.
The commuting control combines the spin-independent form $\mathbf V=U\otimes\id_d$ with a $Z$-only byproduct subset, so all dressed blocks commute; the generic ensemble obtains non-commuting dressed blocks from label-mixing hopping through $V_{j,k}$ together with the fusion byproducts.

As a complementary signature, we probe the conditional distribution $p(\boldsymbol\beta\,|\,\boldsymbol c)$ via its second moment (often referred to as the output collision probability in random-circuit sampling~\cite{Hangleiter2023Computational})
$\Gamma(\boldsymbol c):=\sum_{{\boldsymbol\beta}}p(\boldsymbol\beta\,|\,\boldsymbol c)^2$.
For a Haar-random state in dimension $d^{2n}$, $\mathbb E[\Gamma]=2/(d^{2n}+1)$~\cite{Collins2006Integration, Harrow2009Random}; we therefore report the Haar-normalized ratio
$(d^{2n}+1)\Gamma/2$ in Fig.~\ref{fig:cp}(a, b).
The monitored-FLO ensemble remains close to this benchmark, while the commuting control shows systematic deviations that grow with $n$.
At the instance level, the rescaled weights $x:=d^{2n}p(\boldsymbol\beta\,|\,\boldsymbol c)$ follow the Porter--Thomas distribution~\cite{Porter1956Fluctuations} in the monitored-FLO ensemble but not in the commuting control (Fig.~\ref{fig:cp}(c, d)), consistent with anti-concentration in the standard sense that a constant fraction of outcomes satisfy $p(\boldsymbol\beta\,|\,\boldsymbol c)=\Omega(d^{-2n})$~\cite{Harrow2009Random, Dalzell2022Random}.

\smallskip
\paragraph*{Hardness conjecture.}
Motivated by the ordered non-commutative structure of $\mathcal T_{\boldsymbol\beta}(\mathbf S)$ (Proposition~\ref{prop:nc-amp}; cf.\ SM for worst-case exact-value benchmarks), by our numerical evidence for strong non-commutativity, rapid growth of the order-respecting MPO bond dimension, and anti-concentration (Figs.~\ref{fig:bond} and~\ref{fig:cp}), and by the dilute collision-free regime $m=\Theta(n^2)$ of the accepted truncated-Haar block ensemble, the counterpart of the regime underlying boson-sampling conjectures~\cite{Aaronson2011The}, we formulate the following average-case hardness conjecture, in the standard Stockmeyer framework for sampling-advantage proposals~\cite{Stockmeyer1983,Aaronson2011The,Aaronson2017Complexity,Hangleiter2023Computational} (assuming a standard finite-precision discretization of the Haar instance so inputs admit $\poly(n)$-bit descriptions).

\begin{figure}[t]
\centering
\includegraphics[width=1.0\columnwidth]{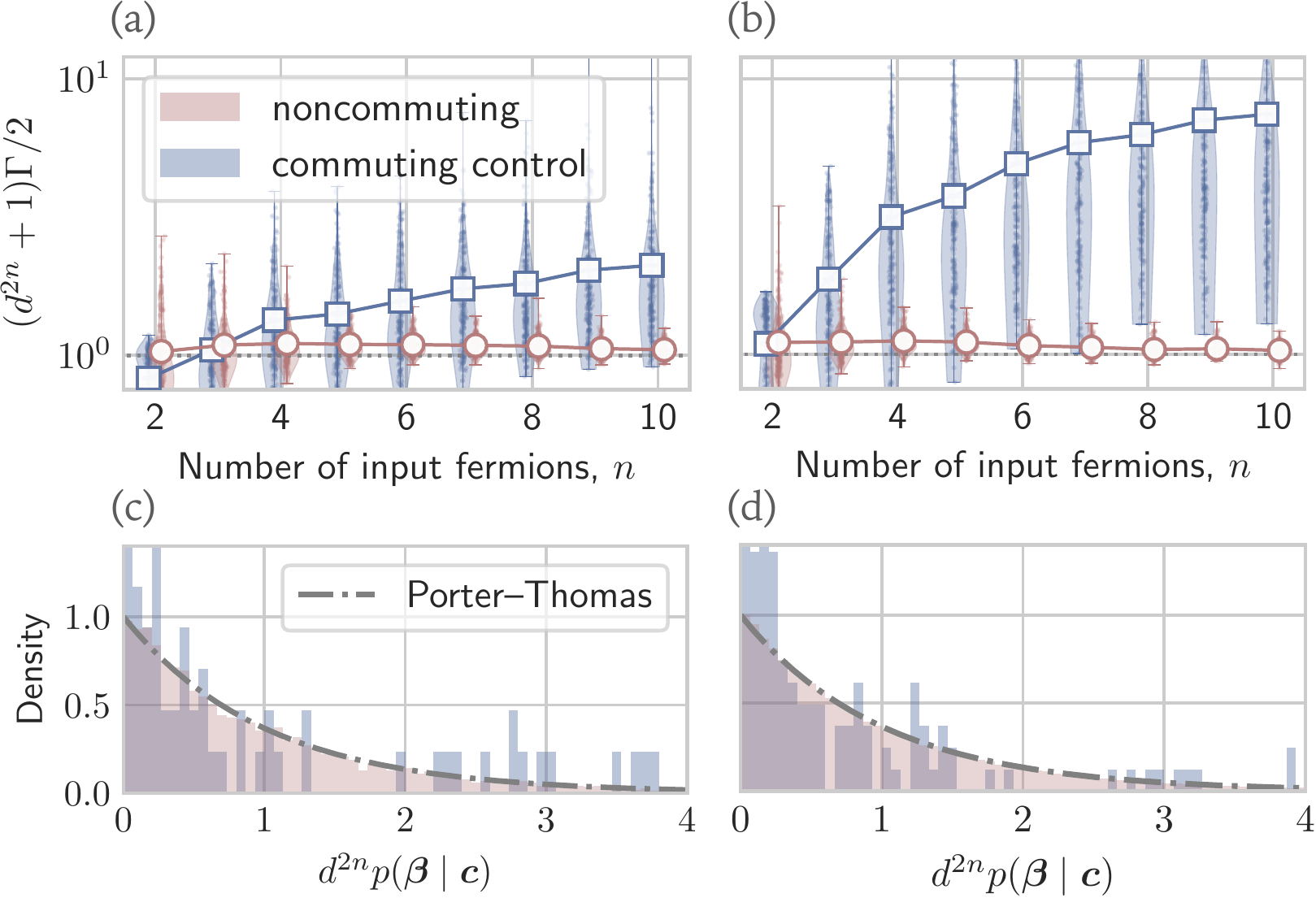}
\caption{
Second-moment ($\Gamma(\boldsymbol c)=\sum_{\boldsymbol\beta}p(\boldsymbol\beta\,|\,\boldsymbol c)^2$) and Porter--Thomas diagnostics of monitored-FLO branches.
(a, b) Haar-normalized second moment versus $n$.
Red: monitored-FLO ensemble.
Blue: commuting-control benchmark ($S_{t,k}\propto \id_d$ with $Z$-only byproducts).
The dotted line marks the Haar random-state average. Accepted-ensemble sampling as in Fig.~\ref{fig:bond} (Born-sampled records, 200 per $n$).
(c, d) Distributions of rescaled branch probabilities $x:=d^{2n}p(\boldsymbol\beta\,|\,\boldsymbol c)$ for representative instances with $n=8$, compared to the Porter--Thomas law $P(x)=e^{-x}$ (grey dash-dotted).
(a, c) $d=2$. (b, d) $d=3$.
}
\label{fig:cp}
\end{figure}

\begin{conjecture}[Average-case hardness of monitored-FLO branches]\label{conj:avg-hard-sampling}
Fix $d\ge2$ and $\kappa>0$, and set \(m=\lceil\kappa n^2\rceil\).
Draw $\mathbf V$ Haar-random and then draw the monitoring record $\boldsymbol c$ from the Born distribution conditioned on collision-free acceptance. Given the accepted pair $(\mathbf V,\boldsymbol c)$, it is $\#\mathrm P$-hard to output estimates of $p(\boldsymbol\beta\,|\,\boldsymbol c)$ achieving relative error $1/\poly(n)$ on a constant fraction of accepted instance--output pairs $((\mathbf V,\boldsymbol c),\boldsymbol\beta)$, with $\boldsymbol\beta$ drawn uniformly from $(\Z_d^2)^n$.
\end{conjecture}

\begin{theorem}[Conditional hardness of approximate sampling]
\label{thm:cond-sampling}
Assume Conjecture~\ref{conj:avg-hard-sampling} and accepted-ensemble anti-concentration, i.e., a constant fraction of accepted instance--output pairs satisfy $p(\boldsymbol\beta\,|\,\boldsymbol c)=\Omega(d^{-2n})$
(cf.\ Fig.~\ref{fig:cp}).
Then no classical probabilistic polynomial-time algorithm, given the classical description of $\mathbf V$ and the realized accepted record $\boldsymbol c$, can sample from $p(\boldsymbol\beta\,|\,\boldsymbol c)$ within inverse-polynomial average total-variation distance over the accepted ensemble, unless the polynomial 
hierarchy collapses.
\end{theorem}
\noindent The SM proves Theorem~\ref{thm:cond-sampling}, giving the precise quantifiers over the accepted instance--record ensemble and a pair-level anti-concentration condition, and extends the hardness to samplers of the full heralded transcript, including the boundary readout, with the assumptions read in the corresponding convention.

\paragraph*{Discussion and outlook.} 
We identify measurement-induced non-commutativity as a route by which monitored free-fermion dynamics escape determinantal/Pfaffian structure. 
The key ingredient is fusion-enforced ordering: mid-circuit monitoring and feedforward pin the contraction order of the antisymmetrized FLO permutation sum, and each accepted branch amplitude $\mathcal A_d(\boldsymbol c,\boldsymbol\beta)=\langle r|\mathcal T_{\boldsymbol\beta}(\mathbf S)|\ell\rangle$ is a matrix element of an ordered non-commutative trace polynomial.
This route is complementary to fermionic advantage schemes based on magic-state injection~\cite{Hebenstreit2019All,Hebenstreit2020Computational,Oszmaniec2022Fermion}: the hard object is the native branch probability of a monitored-FLO ensemble, rather than the amplitude of an interacting circuit synthesized from FLO and a non-Gaussian resource. Free-fermion simulability is therefore contingent on measurement structure, reminiscent of measurement-induced phenomena in monitored dynamics~\cite{Skinner2019Measurement,Alberton2021Entanglement}, there in entanglement structure, here in computational complexity. The architecture therefore realizes a different hardware trade-off: it avoids dedicated fermionic magic-state injection and interacting fermionic propagation, while requiring Bell-pair preparation, mid-circuit occupation readout, Bell fusion, and feedforward. For the qudit-encoded route these are built from standard primitives demonstrated separately on current platforms. For fermion-native realizations the open challenges are a payload-preserving, quantum-nondemolition occupation flag and coherent Bell fusion. The End Matter details both routes.

As a complementary worst-case anchor, the SM relates $\mathcal{T}_{\boldsymbol\beta}(\mathbf S)$ to Cayley's row-ordered determinant~\cite{Arvind2010On,Chien2011Almost,Gentry2014Noncommutative} and exhibits a restricted commuting family on which exact evaluation of the branch amplitude itself, in the equal-boundary convention of step~(v), is $\#\mathrm P$-hard for fixed $d\ge 3$. There, the exact-value hardness is carried by the
measurement-enforced cycle weighting alone; for generic branches, the finite-size numerics separately provide evidence of non-commutativity, anti-concentration, and rapid sequential-memory growth. A natural next question is whether monitored-FLO branches anti-concentrate already at shallow, plausibly logarithmic, depth, as suggested for active FLO circuits~\cite{Oszmaniec2022Fermion}.

\smallskip
\paragraph{Acknowledgments.}
\begin{acknowledgments}
We thank Antonio Anna Mele, Lorenzo Leone, and Nathan Walk for discussions. This work has been supported by the BMFTR (RealistiQ, QSolid, MuniQC-Atoms, QuSol, Hybrid++, FermiQP, PasQuops), the Munich Quantum Valley (K-8), the Quantum Flagship (Millenion, PasQuans2), Berlin Quantum, the European Research Council (DebuQC), the Clusters of Excellence MATH+ and ML4Q, the DFG (CRC 183 and SPP 2514), and the Alexander von Humboldt Foundation.
\end{acknowledgments}

\bibliography{bibliography}

\twocolumngrid
\appendix

\clearpage
\begin{center}
	\textbf{\large End Matter}
\end{center}

\makeatletter

\subsection*{Local encode/decode for qubits}
For $d=2$, the auxiliary register $\widetilde Q_j$ is a single flag qubit $f_j$.
The two-qubit block $(Q_j,f_j)$ has dimension $4$, which contains the local
vacuum-plus-single-particle sector (dimension $3$) plus one extra basis state
treated as outside the code space (a local collision state).
We identify the vacuum and single-particle codewords by
\begin{equation}
|{\rm vac}\rangle_j \leftrightarrow |0,0\rangle_{Q_j ,f_j},\quad
a^\dagger_{j,\alpha}|{\rm vac}\rangle_j \leftrightarrow |1-\alpha,\alpha\rangle_{Q_j, f_j},
\end{equation}
with \(\alpha\in\{0,1\}\), and use the remaining basis state $|1,1\rangle_{Q_j ,f_j}$ as an ``out-of-code'' collision marker.

A convenient constant-depth choice (Fig.~\ref{fig:encdec_d2}) is to apply
$\mathrm{CNOT}_{Q_j\to f_j}$ followed by $X_{Q_j}$, i.e., 
$\mathcal E_j:=X_{Q_j}\mathrm{CNOT}_{Q_j\to f_j}$.
For decoding we apply $\mathrm{CNOT}_{Q_j\to f_j}$ and then measure the flag qubit $f_j$ in the computational basis; we write this as
$\mathcal D_j:=\mathrm{CNOT}_{Q_j\to f_j}$ followed by a computational-basis measurement of $f_j$.
On the single-particle sector, $\mathcal D_j$ maps any valid codeword to $f_j=1$ while preserving the (unresolved) payload
state on $Q_j$ up to a fixed local basis swap, e.g.,
\begin{equation}
\label{eq:endmatter-d2-preserve}
\mathcal D_j\left(c_0|1,0\rangle+c_1|0,1\rangle\right)
=\left(c_0|1\rangle+c_1|0\rangle\right)_{Q_j}\otimes|1\rangle_{f_j}.
\end{equation}
Thus measuring $f_j$ reveals only whether the block lies in the decoded single-particle sector, without resolving
the internal qubit label.

More generally, for arbitrary $d$ the coarse-grained monitoring corresponds to the projector 
\begin{equation} \Pi^{(j)}_{1}=\sum_{\alpha=0}^{d-1} n_{j,\alpha}\prod_{\nu\neq\alpha}\bigl(1-n_{j,\nu}\bigr) \end{equation} 
onto the single-occupation sector.
This operator is diagonal in the occupation basis and acts as the identity on the internal degree of freedom within the single-particle subspace. It lies outside fermionic Gaussian (FLO/matchgate) measurement primitives that resolve individual mode occupations. For $d=2$ this reduces to 
\begin{equation} \Pi^{(j)}_{1}=n_{j,0}(1-n_{j,1})+n_{j,1}(1-n_{j,0}), \end{equation} 
which confirms that conditioning on $f_j=1$ locally annihilates both vacuum and multi-particle components while preserving the quantum coherence of the internal state. 

In the monitored protocol, we set $c_j:=f_j$.
Accepting outcomes with exactly $n$ sites satisfying $c_j=1$ selects $n$ singly occupied blocks; since the
total fermion number is $n$, this excludes collisions.

\begin{figure}[tbp]
\centering
\includegraphics[width=0.7\columnwidth]{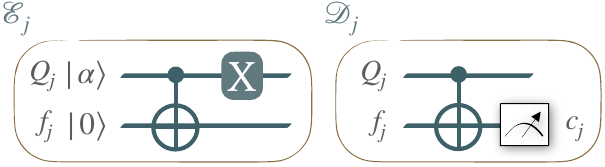}
\caption{
Local $d=2$ encode/decode gadgets on the pair $(Q_j,f_j)$.
Left: the encoder $\mathcal E_j$ maps $|\alpha\rangle_{Q_j}|0\rangle_{f_j}$ to the one-hot codewords
$|1,0\rangle$ ($\alpha=0$) and $|0,1\rangle$ ($\alpha=1$).
Right: the decoder $\mathcal D_j$ applies a local two-qubit gate and measures $f_j$, producing the
coarse-grained occupation flag $c_j$ while leaving the internal state on $Q_j$ unresolved.
}
\label{fig:encdec_d2}
\end{figure}

\subsection*{Physical implementation}
We describe how the three layers of the protocol map onto experimental primitives when the $d$ fermionic modes of each block are encoded into a qubit or qudit register. Throughout, the coherent propagation is strictly number-conserving FLO. Bell-pair preparation, flag extraction, and Bell fusion are local primitives of a fixed preparation-and-readout interface. What leaves Gaussian FLO simulation, namely the hybrid Bell-pair encoding, the coarse occupation readout, and the Bell-fusion measurement, resides in this interface rather than in any coherent matchgate of the FLO layer.

The propagator $\mathbf V\in\mathrm U(dm)$ is a particle-conserving Gaussian (matchgate) transformation. Under the Jordan--Wigner correspondence it factorizes into two-mode hopping/Givens rotations and local phases~\cite{Jozsa2008Matchgates,Oszmaniec2022Fermion}, primitives already implemented on superconducting processors in Hartree--Fock/Givens-rotation and digital Fermi--Hubbard simulations~\cite{Arute2020Hartree,Barends2015Digital}. The internal $d$-level label is carried by a native qudit, such as a multilevel ion, atom, or transmon, or by a small qubit register. During the FLO layer it is coherently mixed as part of the single-particle orbital dynamics in the blocks $V_{j,k}$, while the monitoring never resolves it before the Bell-fusion layer.

The coarse-grained measurement is the single-occupation projector $\Pi^{(j)}_1$ above: the decoder $\mathcal D_j$ maps occupancy onto the auxiliary flag $f_j$, which is read out while the payload on $Q_j$ stays coherent, so monitoring reveals only whether a block is singly occupied and never the label $\alpha$. In the encoded route this is an ordinary ancilla-flag measurement, not a destructive detection of the payload, and is distinct from the harder fermion-native occupation readout discussed below. Spectator-preserving mid-circuit readout, reset, and feedforward have been demonstrated in neutral-atom and trapped-ion processors~\cite{Graham2023Midcircuit,Norcia2023Midcircuit,Lis2023Midcircuit,Iqbal2024Topological}, and superconducting dynamic circuits supply the same mid-circuit measurement and real-time feedforward capability, including teleportation gadgets fed forward across a hundred qubits~\cite{Baumer2024Efficient,CarreraVazquez2024Combining}. The encoded route thus runs on a single universal digital processor: the FLO layer, the flag ancillas, and the Bell fusions are gate-and-measurement primitives, and superconducting processors provide both the Givens-rotation primitives above and dynamic-circuit feedforward in the same platform class. The routing into the fixed fusion geometry compiles into a swap network on any local architecture or into dynamic tweezer rearrangement on reconfigurable atom arrays~\cite{Bluvstein2022A,Bluvstein2024Logical}, with trapped-ion processors providing the same mid-circuit routing primitives~\cite{Postler2022Demonstration}.

A generalized Bell fusion on $(A_{t-1},Q_{\ell_t})$ can be implemented by applying the inverse Bell-pair-preparation circuit and then reading out both systems in the computational basis. For $d=2$ it is a CNOT, a Hadamard on the control, and computational-basis readout of both qubits, constant-depth primitives of the same kind as the encode/decode gadgets. For $d>2$ it is a controlled-shift gate and a qudit Fourier transform with qudit readout; native two-qudit entangling gates up to dimension five have been demonstrated in trapped ions~\cite{Hrmo2023Native}, with a single trapped-ion qudit processor reaching local dimension up to seven~\cite{Ringbauer2022Universal}, and qutrit gates in superconducting circuits~\cite{Morvan2021Qutrit,Goss2022High}, so for fixed $d$ each fusion has only constant local overhead, comprising a two-qudit entangling operation and a qudit measurement.

A complementary route encodes the modes directly in fermionic atoms. Fermionic quantum-gas microscopes and optical-tweezer arrays provide site-resolved occupation readout and tunable single-particle Hubbard dynamics~\cite{Cheuk2015Quantum,Spar2022Realization}, and programmable neutral-atom architectures for fermionic processing have been proposed~\cite{Gonzalez2023Fermionic}. There the FLO layer is hardware-native and the internal label is a hyperfine or nuclear-spin state. Spin-independent tunneling $\mathbf V=U\otimes\id_d$ makes every selected block proportional to the identity; the commuting control of Figs.~\ref{fig:bond} and~\ref{fig:cp} additionally restricts the byproducts to the abelian $Z$-only subgroup, since full Weyl byproducts can already generate non-commutativity for spin-independent tunneling. Label-mixing (spin-dependent) tunneling supplies the additional generic block-level source of non-commutativity. A payload-preserving, quantum-nondemolition occupation flag and coherent Bell fusion remain the main integration challenges.

\begin{figure}[tbh]
\centering
\includegraphics[width=0.8\columnwidth]{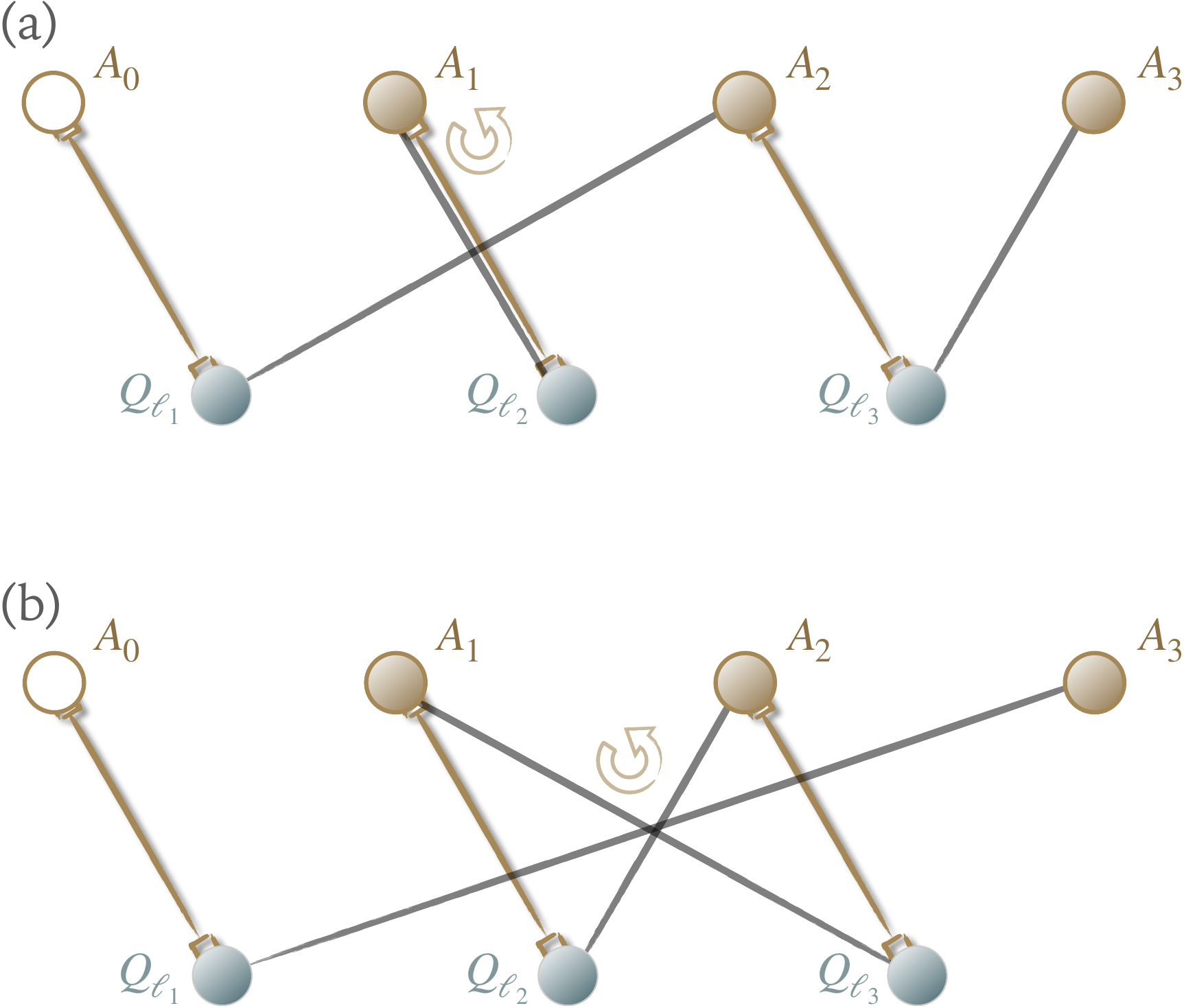}
\caption{
Representative path--cycle decompositions for $n=3$.
Gold edges are the fixed Bell projections $(A_{t-1},Q_{\ell_t})$; black wiring is the permutation tensor $\mathcal P_{\sigma^{-1}}$, so together they give directed edges $A_{t-1}\to A_{\sigma(t)}$ labeled by $\widetilde S_{t,\sigma(t)}$.
(a) $\sigma=(2,1,3)$: a one-cycle at $A_1$ gives ${\Tr}(\widetilde S_{2,1})$, with open path $\widetilde S_{3,3}\widetilde S_{1,2}$.
(b) $\sigma=(3,2,1)$: a two-cycle $A_1\leftrightarrow A_2$ gives ${\Tr}(\widetilde S_{3,1}\widetilde S_{2,2})$, with open path $\widetilde S_{1,3}$.
}
\label{fig:APloop}
\end{figure}

\subsection*{Branch polynomials from the fusion geometry}
We show how the fixed fusion geometry turns the permutation sum in $\det_{\otimes}(\mathbf S)$ into an ordered non-commutative trace polynomial.
Fix an accepted record, hence an $n\times n$ block matrix $\mathbf S=(S_{t,k})_{t,k=1}^n$, and fusion outcomes ${\boldsymbol\beta}=(\beta_1,\dots,\beta_n)$.
The fusion pairs are ordered as $(A_{t-1},Q_{\ell_t})$ for $t=1,\dots,n$, with boundary input $A_0$ and output $A_n$, and we use the dressed blocks
\begin{equation}
\label{eq:endmatter-dressed}
\widetilde S_{t,k}:=S_{t,k}^{\top}\hat D_{\beta_t}.
\end{equation}
Each Bell-projection branch contributes a fixed scalar $d^{-1}$; these factors are kept below and cancel in the normalized law $p(\boldsymbol\beta\,|\,\boldsymbol c)$.
For one permutation term $\sigma$ in $\det_{\otimes}(\mathbf S)$ [Eq.~\eqref{eq:ftpd-maintext}], composing the permutation wiring with the ordered Bell projections induces a directed graph on $\{A_0,\dots,A_n\}$ with labeled edges
\begin{equation}
A_{t-1}\xrightarrow{\ \widetilde S_{t,\sigma(t)}\ } A_{\sigma(t)},
\qquad t=1,\dots,n .
\end{equation}
Since $\sigma$ is a permutation, this graph decomposes uniquely into one open path from $A_0$ to $A_n$ and disjoint directed cycles: the path gives an ordered matrix product with the edge entering $A_n$ as the leftmost factor, each cycle a scalar trace loop. Hence
\begin{equation}
\label{eq:endmatter-pathcycle-eval}
\mathcal T_\sigma
=
d^{-n}
\Bigl(\prod_{C\in\mathrm{Cyc}(\sigma)}\Tr(\Pi_C)\Bigr)
\Pi_{\mathrm{Path}(\sigma)},
\end{equation}
where $\Pi_{\mathrm{Path}(\sigma)}$ and $\Pi_C$ are the ordered products of labels along the open path and around the cycle $C$, and the branch operator is $\mathcal T_{\boldsymbol\beta}(\mathbf S)=\sum_{\sigma\in\mathfrak S_n}\sgn(\sigma)\mathcal T_\sigma$. For $n=2$,
\begin{equation}
\label{eq:endmatter-n2}
\mathcal{T}_{\boldsymbol\beta}(\mathbf S)
=\frac{1}{d^2}\Bigl(
\widetilde S_{2,2}\widetilde S_{1,1}
-\Tr(\widetilde S_{2,1})\,\widetilde S_{1,2}
\Bigr),
\end{equation}
where the identity yields the open path $\widetilde S_{2,2}\widetilde S_{1,1}$ and the swap a one-cycle at $A_1$ with trace factor $\Tr(\widetilde S_{2,1})$. Fig.~\ref{fig:APloop} shows two representative $n=3$ terms, a one-cycle and a two-cycle.

For $n=3$, applying Eq.~\eqref{eq:endmatter-pathcycle-eval} to all six permutations gives
\begin{equation}
\label{eq:endmatter-n3-terms}
\begin{aligned}
d^3\mathcal{T}_{(1,2,3)}&=\widetilde S_{3,3}\widetilde S_{2,2}\widetilde S_{1,1}, \\
d^3\mathcal{T}_{(2,1,3)}&=\Tr(\widetilde S_{2,1})\,\widetilde S_{3,3}\widetilde S_{1,2},\\
d^3\mathcal{T}_{(1,3,2)}&=\Tr(\widetilde S_{3,2})\,\widetilde S_{2,3}\widetilde S_{1,1}, \\
d^3\mathcal{T}_{(3,2,1)}&=\Tr(\widetilde S_{3,1}\widetilde S_{2,2})\,\widetilde S_{1,3},\\
d^3\mathcal{T}_{(2,3,1)}&=\widetilde S_{2,3}\widetilde S_{3,1}\widetilde S_{1,2}, \\
d^3\mathcal{T}_{(3,1,2)}&=\Tr(\widetilde S_{2,1})\Tr(\widetilde S_{3,2})\,\widetilde S_{1,3},
\end{aligned}
\end{equation}
and $\mathcal{T}_{\boldsymbol\beta}(\mathbf S)=\sum_{\sigma\in\mathfrak S_3}\sgn(\sigma)\mathcal{T}_\sigma$ is their signed sum.

\clearpage

\onecolumngrid
\appendix

\clearpage
\begin{center}
  \textbf{\large Supplemental Material for ``Measurement-induced non-commutativity in adaptive fermionic linear optics"}
\end{center}

\setcounter{equation}{0}
\setcounter{figure}{0}
\setcounter{table}{0}
\setcounter{page}{1}

\makeatletter
\renewcommand{\thetable}{S\arabic{table}}
\renewcommand{\thefigure}{S\arabic{figure}}
\renewcommand{\theequation}{S\arabic{equation}}
\makeatother

\section*{Appendix A: Proof of Lemma~1 (Constant-rate collision-free monitoring)}
\begin{proof}[Proof of Lemma~1]
We interpret $p_{\rm e}(n,m,d)$ as the probability that Step~(iii) produces a
collision-free record with exactly $n$ effective blocks, where the probability
is taken jointly over the random choice of the single-particle propagator
$\mathbf V$ and the Born rule. By Haar invariance, the averaged collision-free probability depends only on the dimension of the monitored subspace and is independent of the particular $n$-fermion input state (pure or mixed) supported on the encoded sector: writing $U(\mathbf V)$ for the induced propagator on $\wedge^n(\C^{dm})$, which is an irreducible representation of $\mathrm U(dm)$, Schur's lemma gives $\mathbb E_{\mathbf V}\,U(\mathbf V)\,\rho\,U(\mathbf V)^\dagger=\binom{dm}{n}^{-1}\id$ for any input density matrix $\rho$ on this sector, including the reduced state of an input entangled with the auxiliary registers, and hence $\mathbb E_{\mathbf V}\Tr\bigl[\Pi_{\rm cf}\,U(\mathbf V)\rho\,U(\mathbf V)^\dagger\bigr]=\rank(\Pi_{\rm cf})\binom{dm}{n}^{-1}$ for the projector $\Pi_{\rm cf}$ onto the collision-free subspace. 
For concreteness we compute this constant for a fixed basis occupation pattern $I$.

Let the $dm$ fermionic modes be indexed by pairs $(j,\alpha)$ with
$j\in\{1,\dots,m\}$ and $\alpha\in\{0,\dots,d-1\}$. Fix an $n$-fermion mode-occupation pattern $I\subseteq\{1,2,\dots,dm\}$ with $|I|=n$, obtained by encoding the occupied input blocks $\mathcal I_{\rm in}$ (so $I$ consists of one internal mode per block in $\mathcal I_{\rm in}$).
For any $n$-subset $J\subseteq\{1,2,\dots,dm\}$, let $p_J(\mathbf V)$ denote the output
probability of observing $J$ in the occupation-number measurement, i.e., 
\begin{equation}
p_J(\mathbf V)=\bigl|\det(\mathbf V_{J,I})\bigr|^2,
\end{equation}
where $\mathbf V_{J,I}$ is the $n\times n$ sub-matrix of $\mathbf V$ with rows
indexed by $J$ and columns indexed by $I$.

Assume $\mathbf V$ is Haar-random in $\mathrm U(dm)$ and retain the $n$ columns corresponding to the occupied input modes; these columns form a uniformly random orthonormal $n$-frame in $\mathbb C^{dm}$. By left-invariance of the
Haar measure, for any permutation matrix $P$ acting on the output modes,
$\mathbf V$ and $P\mathbf V$ are identically distributed. Since $P$ acts only
by permuting output rows, we have $p_{PJ}(\mathbf V)=p_J(P^{-1}\mathbf V)$, and
hence
\begin{equation}
\mathbb E_{\mathbf V}\bigl[p_{PJ}(\mathbf V)\bigr]
=\mathbb E_{\mathbf V}\bigl[p_J(\mathbf V)\bigr].
\end{equation}
Because the permutation group acts transitively on $n$-subsets $J\subseteq\{1,2,\dots,dm\}$,
it follows that $\mathbb E_{\mathbf V}[p_J(\mathbf V)]$ is the same for all
$J$ with $|J|=n$. Using $\sum_{|J|=n} p_J(\mathbf V)=1$ for every $\mathbf V$,
we conclude
\begin{equation}
\mathbb E_{\mathbf V}\bigl[p_J(\mathbf V)\bigr]=\frac{1}{\binom{dm}{n}}
\quad \text{for all } J\subseteq\{1,2,\dots,dm\},\ |J|=n.
\end{equation}

We now relate the coarse-grained collision-free record $\boldsymbol c$ used in the protocol
to fine-grained mode outcomes $J\subseteq\{1,\dots,dm\}$.
Conditioning on ``collision-free'' means that exactly one mode is occupied in each of $n$ distinct spatial blocks,
without resolving the internal label $\alpha$; equivalently, it projects onto the direct sum of Fock basis outcomes
with one occupied mode per selected block. Since these Fock outcomes are orthogonal, the Born probability of
a collision-free record is the sum of the corresponding fine-grained probabilities $p_J(\mathbf V)$.
Let $\mathcal{J}_{\rm cf}$ be the set of collision-free mode patterns, i.e.,
those $J$ that occupy $n$ distinct spatial blocks (at most one mode per block).
Counting gives
\begin{equation}
|\mathcal{J}_{\rm cf}|=\binom{m}{n} d^n,
\end{equation}
since one chooses the $n$ occupied blocks and then one of the $d$ internal
modes in each occupied block. Therefore
\begin{equation}
p_{\rm e}(n,m,d)
=\mathbb E_{\mathbf V}\Bigl[\sum_{J\in\mathcal{J}_{\rm cf}} p_J(\mathbf V)\Bigr]
=\frac{\binom{m}{n}d^n}{\binom{dm}{n}}
= d^n \frac{m!(dm-n)!}{(m-n)!(dm)!} 
=\prod_{i=0}^{n-1}\frac{1-\frac{i}{m}}{1-\frac{i}{dm}}.
\label{eq:pe-comb}
\end{equation}

\begin{figure}[tbh]
\centering
\includegraphics[width=0.82\columnwidth]{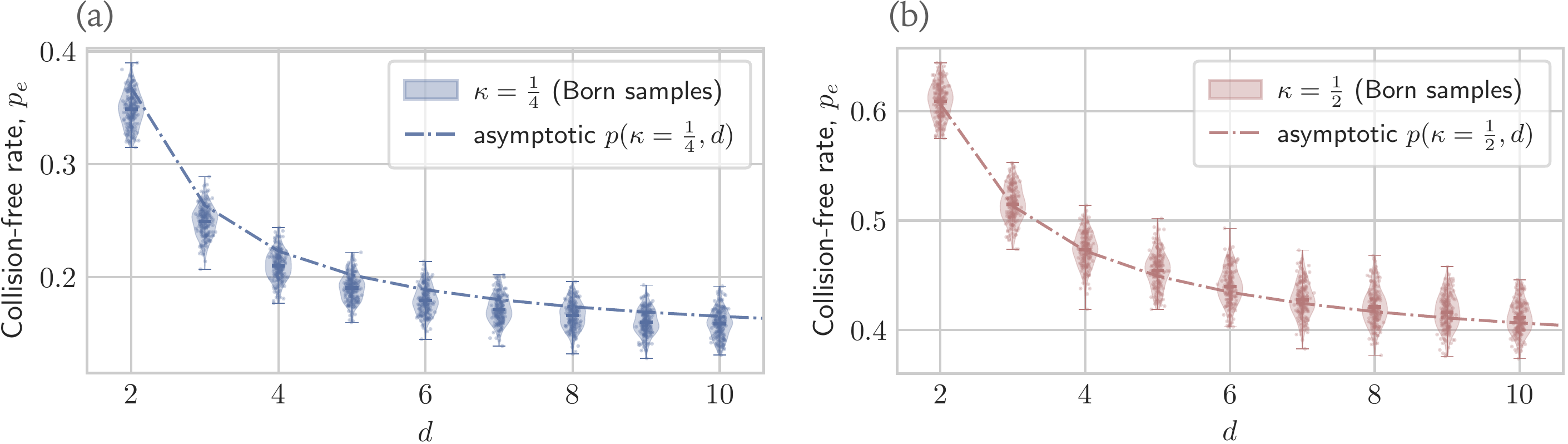}
\caption{Numerical verification of the constant-rate collision-free monitoring probability.
For fixed dilute scaling $m=\lceil\kappa n^2\rceil$ (here $n = 20$),
we estimate the collision-free rate $p_{\rm e}(n,m,d)$ by sampling monitoring records from the Born rule induced by Haar-random
single-particle propagators $\mathbf V\in \mathrm U(dm)$ and recording the collision-free fraction (see text).
Violin plots show the per-instance distribution of empirical rates obtained from repeated Born samples for each independent $\mathbf V$ draw;
horizontal bars indicate medians.
Panels show two dilute scalings: (a) $\kappa=1/4$ (blue) and (b) $\kappa=1/2$ (red).
Dash-dotted curves show the asymptotic prediction $p_\infty(\kappa,d)=\exp\big[-(1-1/d)/(2\kappa)\big]$ from Eq.~\eqref{eq:pe-pinf} for the corresponding $\kappa$.
}
\label{fig:pe_vs_d}
\end{figure}

Now set $m=\lceil\kappa n^2\rceil$, abbreviating $p_{\rm e}(n,\kappa n^2,d):=p_{\rm e}(n,m,d)$. Taking logarithms and using
$\log(1-x)=-x+\mathcal O(x^2)$ uniformly for $x=o(1)$ yields
\begin{align}
\log p_{\rm e}(n,m,d)
&=\sum_{i=0}^{n-1}\Bigl[\log \Bigl(1-\frac{i}{m}\Bigr)
-\log \Bigl(1-\frac{i}{dm}\Bigr)\Bigr] \nonumber\\
&= -\Bigl(1-\frac1d\Bigr)\frac{1}{m}\sum_{i=0}^{n-1} i
\ +\mathcal O \Bigl(\frac{1}{m^2}\sum_{i=0}^{n-1} i^2\Bigr) \nonumber\\
&= -\frac{d-1}{2d\kappa}\ +\mathcal O \Bigl(\frac{1}{n}\Bigr).
\end{align}
Hence, for $m=\lceil\kappa n^2\rceil$,
\begin{equation}
\label{eq:pe-limit}
p_{\rm e}(n,\kappa n^2,d)
=\exp\Bigl(-\frac{d-1}{2d\kappa}\Bigr)\bigl(1+\mathcal O(n^{-1})\bigr),
\end{equation}
and therefore
\begin{equation}
\label{eq:pe-pinf}
p_\infty(\kappa,d):=\lim_{n\to\infty}p_{\rm e}(n,\kappa n^2,d)
=\exp\Bigl(-\frac{1-\frac{1}{d}}{2\kappa}\Bigr).
\end{equation}
In particular, there exists a constant $c_{\kappa,d}>0$ independent of $n$ such that
$p_{\rm e}(n,\kappa n^2,d)\ge c_{\kappa,d}$ for all sufficiently large $n$.
\end{proof}

Let $\rho_{\rm in}$ denote the reduced fermionic input state after Bell initialization and encoding, and let $p_{\mathcal C}(\mathbf V):=\Tr\bigl[\Pi_{\rm cf}\,U(\mathbf V)\,\rho_{\rm in}\,U(\mathbf V)^\dagger\bigr]\in[0,1]$ be the per-instance collision-free probability of the protocol; the fixed-pattern sum $\sum_{J\in\mathcal{J}_{\rm cf}}p_J(\mathbf V)$ was used above only to evaluate the Haar mean, which by the Schur's-lemma argument is input-state independent, so $p_{\rm e}(n,m,d)=\mathbb E_{\mathbf V}\,p_{\mathcal C}(\mathbf V)$. With $\mu:=p_{\rm e}(n,\lceil\kappa n^2\rceil,d)\ge c_{\kappa,d}$, the bound $p_{\mathcal C}(\mathbf V)\le1$ forces $\Pr_{\mathbf V}[p_{\mathcal C}(\mathbf V)\ge\mu/2]\ge\mu/(2-\mu)$, so a constant fraction of Haar instances have constant acceptance probability and the accepted-ensemble statements below rest on a non-negligible instance set rather than an exponentially rare one. The preassigned-record estimate $\binom{m}{n}^{-1}=\exp[-n\log n+O(n)]$ (Appendix~E) uses $m=\lceil\kappa n^2\rceil$. Lemma~1 gives a constant \emph{mean} heralding rate and a constant \emph{fraction} of instances with constant heralding probability, not a high-probability statement; the numerics indicate concentration at the sizes studied.

\noindent\emph{Numerical check.}
Fig.~\ref{fig:pe_vs_d} compares the empirical collision-free monitoring probability $p_{\rm e}(n,\kappa n^2,d)$ against the asymptotic constant
$p_\infty(\kappa,d)$ predicted by Eq.~\eqref{eq:pe-pinf} across several local dimensions $d$ at fixed $\kappa$ (with $m=\lceil \kappa n^2\rceil$).
In the main text we focus on the experimentally most relevant qubit and qutrit cases ($d=2,3$);
here we extend the numerical check to larger $d$ to validate the predicted $d$-dependence of the collision-free rate.
The observed agreement with $p_\infty(\kappa,d)$ and the narrowing spread across instances are consistent with Lemma~1.

\section*{Appendix B: Proof of Lemma~2 (FLO induces the determinantal kernel)}
\begin{proof}[Proof of Lemma~2]
Let the $n$ occupied input blocks be $i_k=k$ for $k=1,\dots,n$ and let the collision-free monitoring record $\boldsymbol c$ select output
blocks $\ell_1<\cdots<\ell_n$ (i.e., $c_{\ell_j}=1$). Recall that the
single-particle propagator $\mathbf V\in \mathrm U(dm)$ is viewed as an $m\times m$
block matrix $\mathbf V=(V_{j,k})$ with $V_{j,k}\in\M_d(\C)$, and the selected block sub-matrix is $\mathbf S=(S_{t,k})_{t,k=1}^n$ with
$S_{t,k}:=V_{\ell_t, i_k}\in \M_d(\C)$.

On the vacuum-plus-single-particle sector, the local encode/decode maps provide
an identification between the internal qudit basis and the creation operators:
for each block index $q$ and internal label $\alpha\in\{0,\dots,d-1\}$,
\begin{equation}
|\alpha\rangle_{Q_q}\ \longleftrightarrow\ a^\dagger_{q,\alpha} |{\rm vac}\rangle .
\end{equation}
Thus an arbitrary input internal state vector
$|\psi\rangle=\sum_{\boldsymbol\alpha}\psi_{\boldsymbol\alpha} 
|\alpha_1,\cdots,\alpha_n\rangle\in(\C^d)^{\otimes n}$
is mapped by the encoding to the $n$-fermion state vector
\begin{equation}
|\Psi_{\rm in}\rangle
=\sum_{\boldsymbol\alpha}\psi_{\boldsymbol\alpha} 
a^\dagger_{i_1,\alpha_1}\cdots a^\dagger_{i_n,\alpha_n} |{\rm vac}\rangle
\ \in\ \wedge^n(\C^{dm}).
\end{equation}
Let $U(\mathbf V)$ be the second-quantized FLO unitary with single-particle
propagator $\mathbf V$. It acts on creation operators as
\begin{equation}
\label{eq:2ndq}
U(\mathbf V) a^\dagger_{k,\alpha} U(\mathbf V)^\dagger
=\sum_{j=1}^m\sum_{\nu=0}^{d-1}(V_{j,k})_{\nu,\alpha} a^\dagger_{j,\nu}.
\end{equation}
Applying \eqref{eq:2ndq} to each factor and expanding, we obtain
\begin{align}
U(\mathbf V)|\Psi_{\rm in}\rangle
=&\sum_{\boldsymbol\alpha}\psi_{\boldsymbol\alpha} 
\prod_{k=1}^n\Bigl(\sum_{j_k,\nu_k}(V_{j_k,i_k})_{\nu_k,\alpha_k} 
a^\dagger_{j_k,\nu_k}\Bigr) |{\rm vac}\rangle \nonumber\\
=&\sum_{\boldsymbol\alpha}\ \sum_{\boldsymbol j}\ \sum_{\boldsymbol\nu}
\psi_{\boldsymbol\alpha} 
\Bigl(\prod_{k=1}^n (V_{j_k,i_k})_{\nu_k,\alpha_k}\Bigr) 
a^\dagger_{j_1,\nu_1} \cdots a^\dagger_{j_{n-1},\nu_{n-1}} a^\dagger_{j_n,\nu_n} |{\rm vac}\rangle .
\label{eq:expand}
\end{align}

Conditioning on the collision-free output sector specified by $\boldsymbol c$ restricts the state to the subspace with exactly one fermion
in each output block $\ell_1,\dots,\ell_n$ and vacuum elsewhere. In the sum
\eqref{eq:expand}, this conditioning eliminates every term except those with
$\{j_1,\dots,j_n\}=\{\ell_1,\dots,\ell_n\}$ (as a set), i.e., those for which
there exists a unique permutation $\sigma\in \mathfrak{S}_n$ such that
\begin{equation}
j_{\sigma(t)}=\ell_t,\qquad t=1,\dots,n.
\end{equation}
For such a term, fermionic anticommutation implies that reordering the creation
operators into the canonical output order $\ell_1<\cdots<\ell_n$ contributes a
sign $\sgn(\sigma)$:
\begin{equation}
a^\dagger_{j_1,\nu_1}\cdots a^\dagger_{j_n,\nu_n}
=\sgn(\sigma)
a^\dagger_{\ell_1,\nu_{\sigma(1)}}\cdots a^\dagger_{\ell_n,\nu_{\sigma(n)}}.
\end{equation}
Substituting $j_{\sigma(t)}=\ell_t$ into the coefficient in \eqref{eq:expand}
yields $\prod_{t=1}^n (V_{\ell_t,i_{\sigma(t)}})_{\nu_{\sigma(t)},\alpha_{\sigma(t)}}$,
i.e., the corresponding block entries of $\mathbf S$.

Finally, applying the local decoding maps on the conditioned single-particle
sector identifies
$a^\dagger_{\ell_t,\nu} |{\rm vac}\rangle\mapsto|\nu\rangle_{Q_{\ell_t}}$
(and writes the occupation flags). Therefore, conditioned on $\boldsymbol c$,
the induced map on the internal $n$-qudit space $(\C^d)^{\otimes n}$ is
\begin{equation}
|\psi\rangle\ \longmapsto\
\sum_{\sigma\in \mathfrak{S}_n}\sgn(\sigma)
\Bigl(\bigotimes_{t=1}^n S_{t,\sigma(t)}\Bigr)\mathcal P_{\sigma^{-1}}|\psi\rangle,
\end{equation}
where we use the convention that $\mathcal P_\sigma$ permutes tensor factors as
\begin{equation}
\label{eq:perm-op-conv}
\mathcal P_\sigma|v_1\otimes\cdots\otimes v_n\rangle
=|v_{\sigma^{-1}(1)}\otimes\cdots\otimes v_{\sigma^{-1}(n)}\rangle.
\end{equation}
With this convention, $\mathcal P_{\sigma^{-1}}$ routes the input tensor factor carried by block $i_{\sigma(t)}$ to output position $t$. This is precisely the operator $\det_{\otimes}(\mathbf S)$ defined in the main text.

\end{proof}

\paragraph*{Implementation remark (compilation of number-conserving FLO).}
A number-conserving FLO unitary $U(\mathbf V)$ is the second-quantized lift of a single-particle transformation
$\mathbf V\in \mathrm U(dm)$ acting on the $dm$ fermionic orbitals. Such unitaries are generated by quadratic, number-conserving
Hamiltonians and can be compiled into two-mode number-conserving gates (fermionic beam splitters / Givens rotations)
together with single-mode phases. Concretely, any $\mathbf V\in \mathrm U(dm)$ admits a decomposition into $\mathcal O(d^2m^2)$ two-mode
Givens rotations, analogous to standard decompositions in linear optics. Under the Jordan--Wigner mapping, each such
two-mode operation corresponds to a matchgate, hence implementable by constant-depth two-qubit primitives on neighboring
qubits after routing. With all-to-all connectivity the resulting circuit depth can be parallelized to $\mathcal O(dm)$, while on local
architectures one can realize the required pairings via a swap network with polynomial overhead.
We refer to Ref.~\cite{Oszmaniec2022Fermion} for explicit constructions and gate counts.

\section*{Appendix C: Bell-projection identity and proof of Proposition~1}
\label{app:nc-amp-proof}

\paragraph*{Convention: Weyl operators and Bell-fusion byproducts.}
We fix the computational basis $\{|k\rangle\}_{k=0}^{d-1}$ and set $\zeta:=e^{2\pi i/d}$.
The standard Weyl (generalized Pauli) generators are
\begin{equation}
X|k\rangle=|k+1\bmod d\rangle,\qquad Z|k\rangle=\zeta^k|k\rangle,
\end{equation}
so that $ZX=\zeta XZ$.
For $\beta=(a,b)\in\Z_d^2$ we use $\hat D_{\beta}:=\hat D_{(a,b)}:=X^{a}Z^{b}$.

\paragraph*{Bell-projection identity.}
We recall the maximally entangled state vector
\begin{equation}
\label{eq:Phi-plus-sm}
|\Phi_d^+\rangle := \frac{1}{\sqrt d}\sum_{k=0}^{d-1}|k,k\rangle
\end{equation}
and define the generalized Bell basis on systems $(1,2)$ by
\begin{equation}
\label{eq:bellstate-def-sm}
|\Phi_{\beta}\rangle_{1,2} := (\id\otimes \hat D_{\beta}^\dagger)|\Phi_d^+\rangle_{1,2}.
\end{equation}
Equivalently,
\begin{equation}
\label{eq:bellbra-def-sm}
\begin{aligned}
\langle\Phi_{\beta}|_{1,2}
=\langle\Phi_d^+|_{1,2}(\id\otimes \hat D_{\beta}), \quad
\langle\Phi_d^+|_{1,2}=\frac{1}{\sqrt d}\sum_{q=0}^{d-1}\langle q|_1\langle q|_2.
\end{aligned}
\end{equation}

\setcounter{lemma}{2}
\begin{lemma}[Three-party Bell-projection identity]
\label{lem:bell-teleport-sm}
Consider three $d$-level systems $(1,2,3)$. For any state vector $|\psi\rangle_1\in\C^d$ and any matrix
$\xi\in\M_d(\C)$ acting on system~$2$,
\begin{equation}
\label{eq:bell-teleport-id-sm}
(\langle\Phi_{\beta}|_{1,2}\otimes\id_3)\Bigl(|\psi\rangle_1\otimes(\xi\otimes\id)|\Phi_d^+\rangle_{23}\Bigr)
= \frac{1}{d}\xi^{\top}\hat D_{\beta}^{\top}|\psi\rangle_3.
\end{equation}
The right-hand side is an unnormalized state vector on system~$3$.
\end{lemma}

\begin{proof}
Write $|\psi\rangle_1=\sum_{r=0}^{d-1}\psi_r|r\rangle_1$. Moreover,
\begin{equation}
\begin{aligned}
(\xi\otimes\id)|\Phi_d^+\rangle_{23}
=\frac{1}{\sqrt d}\sum_{k=0}^{d-1}\Bigl(\sum_{u=0}^{d-1}\xi_{u,k}|u\rangle_2\Bigr)\otimes |k\rangle_3
=\frac{1}{\sqrt d}\sum_{u,k}\xi_{u,k}|u\rangle_2|k\rangle_3.
\end{aligned}
\end{equation}
Using the transpose trick $\langle\Phi_d^+|_{1,2}(\id\otimes \Xi)=\langle\Phi_d^+|_{1,2}(\Xi^{\top}\otimes\id)$ (valid for any $\Xi$), and taking $\Xi=\hat D_{\beta}$, we have
\begin{equation}
\langle\Phi_{\beta}|_{1,2}
=\langle\Phi_d^+|_{1,2}(\id\otimes \hat D_{\beta})
=\langle\Phi_d^+|_{1,2}(\hat D_{\beta}^{\top}\otimes \id).
\end{equation}
Therefore,
\begin{equation}
\begin{aligned}
(\langle\Phi_{\beta}|_{1,2}\otimes\id_3)\Bigl(|\psi\rangle_1\otimes(\xi\otimes\id)|\Phi_d^+\rangle_{23}\Bigr)
=&\frac{1}{d}\sum_{a,b}\sum_{r}\sum_{u,k}
(\hat D_{\beta})_{a b}\psi_r\xi_{u,k}
\langle a|r\rangle\langle b|u\rangle|k\rangle_3 \\
=&\frac{1}{d}\sum_{b,k}\xi_{b,k}\Bigl(\sum_{a}(\hat D_{\beta})_{a b}\psi_a\Bigr)|k\rangle_3 \\
=&\frac{1}{d}\,\xi^{\top}\hat D_{\beta}^{\top}|\psi\rangle_3,
\end{aligned}
\end{equation}
which is~\eqref{eq:bell-teleport-id-sm}.
\end{proof}

\noindent\emph{Remark.}
Lemma~\ref{lem:bell-teleport-sm} naturally produces the byproduct $\hat D_{\beta}^{\top}$.
For the standard choice $X|k\rangle=|k+1\bmod d\rangle$ and $Z|k\rangle=\zeta^k|k\rangle$ with
$\zeta=e^{2\pi i/d}$, one has $X^{\top}=X^{-1}$ and $Z^{\top}=Z$, hence
\begin{equation}
\label{eq:D-transpose-sm}
\begin{aligned}
\hat D_{(a,b)}^{\top}
=(X^a Z^b)^{\top}
= Z^b X^{-a} 
= \zeta^{-ab} X^{-a} Z^b
= \zeta^{-ab}\hat D_{(-a,b)} .
\end{aligned}
\end{equation}
The phase $\zeta^{-ab}$ is a unit-modulus, outcome-dependent scalar, constant within each branch. It can be absorbed into the overall branch prefactor and cancels in the branch law $p(\boldsymbol\beta\,|\,\boldsymbol c)$ and in all diagnostics used below.

Equivalently, using $\hat D_{\beta}^{\top}\propto \hat D_{\overline\beta}$ with $\overline\beta:=(-a,b)$ and absorbing this phase into the branch prefactor, Eq.~\eqref{eq:bell-teleport-id-sm} can be rewritten as
\begin{equation}
(\langle\Phi_{\beta}|_{1,2}\otimes\id_3)\Bigl(|\psi\rangle_1\otimes(\xi\otimes\id)|\Phi_d^+\rangle_{23}\Bigr)
= \frac{1}{d}\xi^{\top}\hat D_{\overline\beta}|\psi\rangle_3.
\end{equation}
In the main text we relabel the measurement outcome $\bar\beta \mapsto \beta$, so fusion step $t$ inserts the unnormalized matrix $d^{-1}S_{t,k}^{\top}\hat D_{\beta_t}$ into the ordered boundary-wire product.

\smallskip
\noindent\emph{Explicit matrices for $d=2,3$.}
For $d=2$ one has $\zeta=-1$ and
\begin{equation}
X=\begin{pmatrix}0&1\\[2pt]1&0\end{pmatrix},\qquad
Z=\begin{pmatrix}1&0\\[2pt]0&-1\end{pmatrix},
\end{equation}
so $\hat D_{(a,b)}=X^aZ^b\in\{\id,X,Z,XZ\}$ and $\hat D_{(a,b)}^{\top}=(-1)^{ab}\hat D_{(a,b)}$.
For $d=3$ one has $\zeta=e^{2\pi i/3}$ and
\begin{equation}
X=\begin{pmatrix}
0&0&1\\
1&0&0\\
0&1&0
\end{pmatrix},\qquad
Z=\mathrm{diag}(1,\zeta,\zeta^2),
\end{equation}
so $\hat D_{(a,b)}=X^aZ^b$ with $a,b\in\{0,1,2\}$ and the transpose induces the nontrivial relabeling $a\mapsto -a$.

\smallskip
\paragraph*{Proof of Proposition~1.}
\begin{proof}
Fix a collision-free monitoring record $\boldsymbol c$, hence the selected block sub-matrix
$\mathbf S=(S_{t,k})_{t,k=1}^n$, and fix fusion outcomes
${\boldsymbol\beta}=(\beta_1,\dots,\beta_n)$.

\emph{(i) FLO kernel.}
Conditioned on $\boldsymbol c$, Lemma~2 yields the $n$-qudit operator kernel
\begin{equation}
\label{eq:det-otimes-sm}
\det_{\otimes}(\mathbf S)
=\sum_{\sigma\in \mathfrak{S}_n}\sgn(\sigma)\Bigl(\bigotimes_{t=1}^n S_{t,\sigma(t)}\Bigr)\mathcal P_{\sigma^{-1}}.
\end{equation}

\emph{(ii) Fixed-order fusion contraction map.}
For fixed $\boldsymbol\beta$, the  Bell-fusion readout defines a fixed tensor network determined by the fusion geometry.
Excluding the boundary vectors $|\ell\rangle$ and $\langle r|$, it induces a $\C$-linear contraction map
\begin{equation}
\label{eq:C-map-sm}
\mathcal C_{\boldsymbol\beta}:\ \M_{d^n}(\C)\ \to\ \M_d(\C),
\end{equation}
independent of $\mathbf S$, such that by definition
\begin{equation}
\label{eq:T-def-sm}
\mathcal{T}_{\boldsymbol\beta}(\mathbf S)\ :=\ \mathcal C_{\boldsymbol\beta}\left(\det_{\otimes}(\mathbf S)\right)\in \M_d(\C).
\end{equation}
The branch amplitude is then obtained as
\begin{equation}
\mathcal A_d(\boldsymbol c,{\boldsymbol\beta})=\langle r|\mathcal{T}_{\boldsymbol\beta}(\mathbf S)|\ell\rangle.
\end{equation}
By construction, $\mathcal C_{\boldsymbol\beta}$ includes the $d^{-1}$ scalar from each Bell projection.

\emph{(iii) Local update rule and dressed blocks.}
Evaluating the contraction sequentially in the order $t=1,\dots,n$,
each fusion step is exactly a three-party contraction of the form in Lemma~\ref{lem:bell-teleport-sm}.
Using the outcome relabeling in~\eqref{eq:D-transpose-sm} to match the main-text convention, the update at step $t$
inserts the ordered map $\xi\mapsto d^{-1}\xi^{\top}\hat D_{\beta_t}$ on the boundary wire.
This motivates the dressed blocks
\begin{equation}
\label{eq:dressed-sm}
\widetilde S_{t,k}:=S_{t,k}^{\top}\hat D_{\beta_t}.
\end{equation}

\emph{(iv) Trace-polynomial (path--cycle) structure.}
Fix a permutation term $\sigma$ in~\eqref{eq:det-otimes-sm}.
Composing the wiring tensor $\mathcal P_{\sigma^{-1}}$ with the fixed fusion geometry induces a directed graph
on the auxiliary vertices $\{A_0,\dots,A_n\}$: for each fusion step $t=1,\dots,n$ there is a directed edge
\begin{equation}
A_{t-1}\xrightarrow{\ \widetilde S_{t,\sigma(t)}\ } A_{\sigma(t)},
\end{equation}
where $\widetilde S_{t,k}=S_{t,k}^{\top}\hat D_{\beta_t}$ are the dressed blocks.
Since $\sigma$ is a permutation of $\{1,\dots,n\}$, every vertex $A_j$ with $j\in\{1,\dots,n\}$ has indegree one,
while $A_0$ has indegree zero and $A_n$ has outdegree zero. Consequently, the graph decomposes uniquely into a single
open directed path from $A_0$ to $A_n$ together with a disjoint union of directed cycles.

Evaluating the tensor network in the fixed pairing order therefore yields a path--cycle form:
the open path from $A_0$ to $A_n$ contributes an ordered matrix monomial $M_\sigma\in\M_d(\C)$, written with the edge entering $A_n$ as the leftmost factor and the edge leaving $A_0$ as the rightmost, and each directed cycle $C$ contributes a scalar loop factor $\Tr(\Pi_C)$ given by the ordered product $\Pi_C$ of the labels around that cycle. Hence the $\sigma$-term evaluates (up to the fixed scalar $d^{-n}$ already absorbed in $\mathcal C_{\boldsymbol\beta}$) to
\begin{equation}
\mathcal{T}_\sigma(\mathbf S)=\Bigl(\prod_{C}\Tr(\Pi_C)\Bigr) M_\sigma.
\end{equation}
Summing over $\sigma\in \mathfrak{S}_n$ with signs and using linearity gives
\begin{equation}
\mathcal{T}_{\boldsymbol\beta}(\mathbf S)\in \M_d(\C)\langle \widetilde S_{t,k}:\ 1\le t,k\le n\rangle_{\Tr}.
\end{equation}
In particular, by the boundary projection we have
$\mathcal A_d(\boldsymbol c,{\boldsymbol\beta})=\langle r|\mathcal{T}_{\boldsymbol\beta}(\mathbf S)|\ell\rangle$.

\end{proof}

\section*{Appendix D: Proof of Theorem~1 (Exponential sequential-memory barrier)}\label{app:memory-barrier}

This appendix formalizes the symbolic ``generic non-commutative'' benchmark used in the main text and proves
Theorem~1.
We keep the same fixed fusion geometry in the protocol
but we treat the dressed blocks as algebraically independent non-commuting generators.
The resulting branch operator is a matrix-valued non-commutative trace polynomial, and we show that any
exact order-respecting sequential contraction requires exponential intermediate memory (MPO bond dimension). Order-respecting MPO contractions can be viewed (up to constant factors depending on $d$) as ordered non-commutative algebraic branching programs computing the corresponding trace polynomial; Theorem~1 follows via Nisan’s rank method~\cite{Nisan1991Lower} adapted to the trace-indeterminate extension.

\paragraph*{Free trace-polynomial model.}
Fix $n$ and $d\ge 2$ and a branch $(\boldsymbol c,\boldsymbol\beta)$.
As in the main text, denote the dressed blocks by $\widetilde S_{t,k}:=S_{t,k}^{\top}\hat D_{\beta_t}$.
In the symbolic benchmark we interpret each $\widetilde S_{t,k}$ as a formal non-commuting generator.

Closed directed wiring cycles contribute scalar trace-loop factors (cf.\ End Matter).
To model these 
loop contributions while remembering which generators appear in the 
loop, we adjoin central commuting
indeterminates indexed by 
cyclic words.
Let $\mathrm{Cyc}(\widetilde S)$ denote the set of cyclic equivalence classes of \emph{non-empty} words in the alphabet
$\{\widetilde S_{t,k}\}_{t,k=1}^n$.
For each $\gamma\in \mathrm{Cyc}(\widetilde S)$ introduce a central commuting indeterminate $\tau_\gamma$, and define
\begin{equation}
\C\langle \widetilde S\rangle_{\Tr}
:= \C\langle \widetilde S_{t,k}:\ t,k\in\{1,2,\ldots,n\}\rangle \ \otimes\ \C[\tau_\gamma:\gamma\in\mathrm{Cyc}(\widetilde S)].
\end{equation}
We interpret each loop factor produced by a directed cycle $C$ as the corresponding $\tau_{\gamma(C)}$, where
$\gamma(C)$ is the cyclic word read along that cycle.

This benchmark is formulated in the free (non-commutative) trace-polynomial algebra and does not impose matrix polynomial identities specific to fixed dimension $d$ (e.g., Cayley--Hamilton).

\paragraph*{Evaluation convention for trace indeterminates.}
Whenever we specialize the non-commuting generators $\widetilde S_{t,k}$ to concrete matrices in $\M_d(\C)$,
we evaluate the trace indeterminates by the (unnormalized) trace of the corresponding matrix word:
for any assignment $\mathrm{ev}$ of all $\widetilde S_{t,k}$ to matrices, extend $\mathrm{ev}$ to
$\C\langle \widetilde S\rangle_{\Tr}$ by
\begin{equation}
\mathrm{ev}(\tau_\gamma):=\Tr\left(\mathrm{ev}(w_\gamma)\right),
\end{equation}
where $w_\gamma$ is any representative word of the cyclic class $\gamma$.
This is well-defined by cyclicity of the trace, and ensures in particular that if any generator in $w_\gamma$ evaluates to $0$,
then $\mathrm{ev}(\tau_\gamma)=0$.
With this convention,
\begin{equation}
\mathcal{T}_{\boldsymbol\beta}(\mathbf S)\in \M_d(\C)\langle \widetilde S\rangle_{\Tr}
\end{equation}
exactly as in Proposition~1, now interpreted purely symbolically (no matrix trace identities imposed).

\paragraph*{Order-respecting contractions and bond dimension.}
An order-respecting sequential contraction processes fusion steps in the enforced order $t=1,2,\dots,n$.
Across the cut between steps $t$ and $t{+}1$, the virtual index must encode the open auxiliary wirings
(i.e., the collection of open strands that still need to be merged/closed by later fusion steps).
The local update at step $t$ is linear in the row-$t$ generators $\{\widetilde S_{t,k}\}_{k=1}^n$,
but it acts on a virtual open-wiring space that may carry multiple open strands.
Merging two strands corresponds to the bilinear multiplication map $m(A\otimes B)=AB$ (linearized via tensor products),
and closing a directed cycle yields a commuting scalar loop factor $\tau_\gamma$.
This is precisely the mechanism that allows the final monomials in $\mathcal{T}_{\boldsymbol\beta}$ to have
non-monotone step labels (cf.\ End Matter).

For cut $t$, let $\chi_t$ denote the minimal possible virtual dimension across that cut among all exact
order-respecting sequential contractions, equivalently (up to constant factors depending only on $d$), the minimal width
at layer $t$ of an ordered non-commutative algebraic branching program computing $\mathcal{T}_{\boldsymbol\beta}$ in the
above free trace-polynomial model.

\paragraph*{Proof of Theorem~1.}
\begin{proof}
Fix a cut $t\in\{0,1,\dots,n\}$.
We lower bound $\chi_t$ by exhibiting $\binom{n}{t}$ linearly independent ``used-leg'' sectors that any
order-respecting sequential contraction must distinguish after $t$ steps.

In each permutation term of the FLO kernel $\det_{\otimes}(\mathbf S)$, fusion step $s$ selects exactly one input-leg label
$k=\sigma(s)$, and each label is used exactly once.
Hence after the first $t$ fusion steps, the set of already used input legs, $\{\sigma(1),\dots,\sigma(t)\}$, is a subset of $\{1,2,\ldots,n\}$ of size $t$. There are $\binom{n}{t}$ such possible subsets, which we will index by $I$.

Since we work in the free trace-polynomial model, the standard basis elements, each a word in the non-commuting generators multiplied by a monomial in the trace indeterminates, are linearly independent.

Now form a rank witness matrix $M$ indexed by $t$-subsets.
For each $t$-subset $I=\{i_1<\cdots<i_t\}\subseteq\{1,2,\ldots,n\}$ define a prefix evaluation $\mathrm{ev}^{\le t}_I$ on the generators $\{\widetilde S_{s,k}\}$ by
\begin{equation}
\widetilde S_{s,k}\mapsto
\begin{cases}
\id_d,& s\le t,\ k=i_s,\\
0,& s\le t,\ k\neq i_s,
\end{cases}
\end{equation}
and regard $\mathrm{ev}^{\le t}_I$ as a partial assignment of the generators, leaving all generators with $s>t$ unassigned.
For each $t$-subset $J=\{j_1<\cdots<j_t\}$ define a suffix assignment $\mathrm{ev}^{>t}_J$ by first letting
$\{u_1<\cdots<u_{n-t}\}=\{1,2,\ldots,n\}\setminus J$, and setting
\begin{equation}
\widetilde S_{s,k}\mapsto
\begin{cases}
\id_d,& s> t,\ k=u_{s-t},\\
0,& s> t,\ k\neq u_{s-t},
\end{cases}
\end{equation}
leaving all generators with $s\le t$ unassigned.
Let $\mathrm{ev}_{I,J}$ denote the combined full evaluation obtained by applying $(\mathrm{ev}^{\le t}_I,\mathrm{ev}^{>t}_J)$,
and extend $\mathrm{ev}_{I,J}$ to the trace indeterminates by the trace convention above:
$\mathrm{ev}_{I,J}(\tau_\gamma)=\Tr(\mathrm{ev}_{I,J}(w_\gamma))$.

In order to proceed, we
define the $\binom{n}{t}\times\binom{n}{t}$ matrix
\begin{equation}
M_{I,J}:=\Tr\left(\mathrm{ev}_{I,J}\big(\mathcal{T}_{\boldsymbol\beta}(\mathbf S)\big)\right),
\end{equation}
where $\Tr$ is the usual trace on $\M_d(\C)$.

By construction, under $\mathrm{ev}_{I,J}$ each row $s$ has exactly one nonzero dressed block (equal to $\id_d$),
at column $i_s$ for $s\le t$ and at column $u_{s-t}$ for $s>t$.
A permutation term $\sigma$ in $\det_{\otimes}(\mathbf S)$ can survive only if it selects these nonzero blocks in every row,
i.e., only if $\sigma(s)$ equals the designated nonzero column in row $s$ for all $s$.
Such a permutation exists if and only if the designated column choices use each label exactly once, which holds exactly when $I=J$.
Therefore:
\begin{itemize}[leftmargin=*,itemsep=1pt,topsep=2pt]
\item If $I\neq J$, no permutation term survives and $M_{I,J}=0$.
\item If $I=J$, exactly one permutation term survives; every surviving loop word evaluates to $\id_d$ and hence each corresponding trace indeterminate evaluates to $\Tr(\id_d)=d$, so $M_{I,I}\neq 0$.
\end{itemize}
Thus $M$ is diagonal with nonzero diagonal entries, hence
\begin{equation}
\rank(M)=\binom{n}{t}.
\end{equation}

Finally, relate this to an order-respecting sequential contraction.
Fix any exact sequential contraction scheme that processes fusion steps $1,2,\ldots,n$ in order and has bond dimension $\chi_t$
across the cut between steps $t$ and $t{+}1$.
Under the evaluations $\mathrm{ev}_{I,J}$, every local tensor entry becomes a constant $d\times d$ matrix (possibly zero), because it is linear in
$\widetilde S_{s,k}$ and $\widetilde S_{s,k}\in\{0,\id_d\}$ under $\mathrm{ev}_{I,J}$.
After applying $\Tr$ to the final $\M_d(\C)$ output, the resulting scalar factors through the cut space
$(\M_d(\C))^{\chi_t}$, which has dimension $d^2\chi_t$.

Concretely, for each $t$-subset $I$ let
\begin{equation}
A_I=(A_{I,1},\dots,A_{I,\chi_t})\in (\M_d(\C))^{\chi_t}
\end{equation}
be the cut state obtained after contracting the first $t$ steps under
$\mathrm{ev}^{\le t}_I$, and for each $t$-subset $J$ let
\begin{equation}
B_J=(B_{J,1},\dots,B_{J,\chi_t})\in (\M_d(\C))^{\chi_t}
\end{equation}
be the corresponding suffix state under $\mathrm{ev}^{>t}_J$.
Cutting the tensor network between steps $t$ and $t{+}1$ yields
\begin{equation}
M_{I,J}= \sum_{x=1}^{\chi_t}\Tr(B_{J,x}A_{I,x}).
\end{equation}
After vectorizing the matrix entries, there exist vectors $a_I,b_J\in\C^{d^2\chi_t}$ such that
\begin{equation}
M_{I,J}= b_J^{\top} a_I.
\end{equation}
Equivalently, $M=AB^{\top}$ with inner dimension $d^2\chi_t$, and hence
\begin{equation}
\rank(M)\le d^2\chi_t .
\end{equation}
Combining with $\rank(M)=\binom{n}{t}$ gives
\begin{equation}
\chi_t\ge d^{-2}\binom{n}{t}.
\end{equation}
Taking $t=\lfloor n/2\rfloor$ yields $\chi_{\lfloor n/2\rfloor}=2^{\Omega(n)}$ for fixed $d$.
\end{proof}

\paragraph*{Remark (scope beyond the free-algebra model).}
Although Theorem~1 is phrased in the free trace-polynomial benchmark, the rank-witness proof uses only concrete matrix substitutions, together with factorization through the cut space. Hence the same lower bound applies to any exact fusion-order--respecting sequential contraction scheme that is required to work uniformly over matrix-valued block inputs; after uniformly rescaling all witness substitutions by any $0<\lambda<2^{-1/2}$, every combined prefix--suffix block matrix is a strict contraction, while row-wise multilinearity sends $M\mapsto\lambda^n M$ and hence preserves its rank. Taking $\boldsymbol\beta=\mathbf 0$, standard unitary completion in the dilute regime then embeds the scaled witness family into the monitored-FLO physical family. This should be interpreted as a uniform worst-case lower bound on such contraction schemes, not as an instancewise lower bound on the minimal bond dimension of a single fixed numerical branch.

\section*{Appendix E: Accepted-branch ensemble and conditional sampling over random records}
\label{app:accepted-ensemble}

This appendix specifies the probabilistic content of the sampling task: the hardness and anti-concentration statements concern the conditional distribution given a collision-free record obtained at random from the device, not a record fixed in advance. For each Haar-random FLO instance the monitoring record is produced by the Born rule, the run is accepted when the record is collision-free with $n$ occupied blocks, and the realized record is retained as public classical side information. The conditional branch distribution entering Conjecture~1 and Theorem~2 is the law of the subsequent fusion and boundary outcomes given this realized record. Operationally, each accepted run outputs a sample from $p_{\mathbf V,\boldsymbol c}$ for the realized record $\boldsymbol c$. All statements below concern this distribution.

\paragraph*{Accepted-branch ensemble.}
Write
\begin{equation}
\mathcal{C}_{n,m}:=\Bigl\{\boldsymbol c\in\{0,1\}^m:\ \textstyle\sum_{j=1}^m c_j=n\Bigr\}
\end{equation}
for the set of accepted collision-free coarse records. Here $c_j=1$ means that block $j$ is in the decoded single-occupation sector; since the total fermion number is $n$, having exactly $n$ such singly occupied blocks excludes any collision. These coarse occupation records sum over the internal labels resolved in the fine-grained mode patterns $\mathcal J_{\rm cf}$ of Appendix~A, so $|\mathcal{C}_{n,m}|=\binom{m}{n}$ rather than $\binom{m}{n}d^n$. Each $\boldsymbol c\in\mathcal{C}_{n,m}$ fixes selected blocks $\ell_1<\cdots<\ell_n$ and hence the selected block sub-matrix
\begin{equation}
S_{t,k}(\mathbf V,\boldsymbol c)=V_{\ell_t,k}.
\end{equation}
For fixed $\mathbf V$, let $P_{\mathbf V}(\boldsymbol c)$ be the Born probability of the coarse record $\boldsymbol c$, and define the total accepted-record probability
\begin{equation}
p_{\mathcal{C}}(\mathbf V):=P_{\mathbf V}(\mathcal{C}_{n,m})=\sum_{\boldsymbol c\in\mathcal{C}_{n,m}}P_{\mathbf V}(\boldsymbol c).
\end{equation}
Conditioned on acceptance, the record is distributed as
\(P_{\mathbf V}^{\mathcal C}(\boldsymbol c)
:=P_{\mathbf V}(\boldsymbol c)/p_{\mathcal C}(\mathbf V)\) on
\(\mathcal C_{n,m}\). We define the accepted-branch ensemble
\(\mathfrak A_{n,d,\kappa}\) to be the joint law of the random pair
\((\mathbf V,\boldsymbol c)\) generated by
\begin{equation}
m=\lceil \kappa n^2\rceil,\qquad
\mathbf V\sim\mu_{\mathrm{Haar}}\ \text{on}\ \mathrm U(dm),
\qquad
\boldsymbol c\,|\,\mathbf V\sim P_{\mathbf V}^{\mathcal C}.
\end{equation}
Equivalently, for any test function \(F\),
\begin{equation}
\mathbb E_{(\mathbf V,\boldsymbol c)\sim\mathfrak A_{n,d,\kappa}}F(\mathbf V,\boldsymbol c)
=
\mathbb E_{\mathbf V\sim\mu_{\mathrm{Haar}}}
\Biggl[
\sum_{\boldsymbol c\in\mathcal C_{n,m}}
P_{\mathbf V}^{\mathcal C}(\boldsymbol c)\,
F(\mathbf V,\boldsymbol c)
\Biggr].
\end{equation}

The ensemble is instance based: $\mathbf V$ is sampled first, and the accepted record is then Born-sampled from the Born distribution of that fixed instance; the device may be run on a fixed $\mathbf V$ until a collision-free record is heralded (see the Remark below for the single-shot variant with constant expected overhead). The record selects which sub-matrix of the already-fixed $\mathbf V$ enters the branch polynomial rather than determining $\mathbf V$ itself, and one accepted run yields a single public transcript $(\boldsymbol c,y)$. For the full-transcript formulation we take $y=(\boldsymbol\beta,r)$ with $\mathcal Y=(\Z_d^2)^n\times\Z_d$ and $N=|\mathcal Y|=d^{2n+1}$. The numerical diagnostics use fixed boundary vectors $|\ell\rangle,|r\rangle$, a constant-dimension matrix-element slice of the same branch operator with $y=\boldsymbol\beta$, $\mathcal Y=(\Z_d^2)^n$, and $N=d^{2n}$; this does not affect the ordered polynomial structure.

\smallskip
\noindent\emph{Remark (single-shot joint acceptance).} Alternatively, one may draw a fresh $\mathbf V\sim\mu_{\mathrm{Haar}}$ on every run, execute the monitored circuit once, and retain the pair $(\mathbf V,\boldsymbol c)$ only when the record is collision-free. The resulting joint-conditioned law $\widetilde{\mathfrak A}_{n,d,\kappa}$, defined by $\mathrm d\widetilde{\mathfrak A}_{n,d,\kappa}(\mathbf V,\boldsymbol c)=\mu_{\mathrm{Haar}}(\mathrm d\mathbf V)\,P_{\mathbf V}(\boldsymbol c)/\bar p_{\mathcal C}$ for $\boldsymbol c\in\mathcal C_{n,m}$, with $\bar p_{\mathcal C}:=\mathbb E_{\mathbf V}\,p_{\mathcal C}(\mathbf V)=p_{\rm e}(n,m,d)\ge c_{\kappa,d}$ (Lemma~1), is generated at constant expected overhead: an accepted pair is obtained after $1/\bar p_{\mathcal C}=O(1)$ preparation runs in expectation, with no per-instance repetition. Since $p_{\mathcal C}(\mathbf V)\le 1$, every event $A$ of accepted pairs obeys $\bar p_{\mathcal C}\,\widetilde{\mathfrak A}_{n,d,\kappa}(A)\le\mathfrak A_{n,d,\kappa}(A)$; that is, $\widetilde{\mathfrak A}_{n,d,\kappa}$ is absolutely continuous with respect to $\mathfrak A_{n,d,\kappa}$ with density $p_{\mathcal C}(\mathbf V)/\bar p_{\mathcal C}\le 1/\bar p_{\mathcal C}=O(1)$, so an event that is rare under $\mathfrak A_{n,d,\kappa}$ remains rare under the single-shot law, up to the constant factor $1/\bar p_{\mathcal C}$. In particular, an efficient estimator succeeding on a constant fraction of pairs drawn from $\widetilde{\mathfrak A}_{n,d,\kappa}$ succeeds on a constant fraction under $\mathfrak A_{n,d,\kappa}$, so Conjecture~1 equally rules out efficient $1/\poly(n)$-relative-error estimation over the single-shot ensemble. The same density bound holds after adjoining an independent $y\sim\mathrm{Unif}(\mathcal Y)$; hence Appendix~F applies with $\mathfrak A_{n,d,\kappa}$ replaced throughout by $\widetilde{\mathfrak A}_{n,d,\kappa}$, provided Eq.~\eqref{eq:F-ac} is imposed under the latter ensemble.

The record $\boldsymbol c$ is the conditioning variable, observed and made public and used to route the feedforward, while $y$ is the output. For a realized accepted pair the unnormalized branch weight is
\begin{equation}
w_{\mathbf V,\boldsymbol c}(y):=\bigl|\langle r|\mathcal T_{\boldsymbol\beta}(\mathbf S(\mathbf V,\boldsymbol c))|\ell\rangle\bigr|^2=\bigl|\mathcal A_d(\boldsymbol c,\boldsymbol\beta)\bigr|^2,
\end{equation}
and the normalized conditional law, denoted $p(\boldsymbol\beta\,|\,\boldsymbol c)$ in the main text, is $p_{\mathbf V,\boldsymbol c}(y):=w_{\mathbf V,\boldsymbol c}(y)/\sum_{y'}w_{\mathbf V,\boldsymbol c}(y')$. Thus $w_{\mathbf V,\boldsymbol c}$ is not a joint probability for $(\boldsymbol c,y)$; in the full-transcript convention $y=(\boldsymbol\beta,r)$, the accepted transcript law is the disintegration $\Pi_{\mathbf V}(\boldsymbol c,y)=P_{\mathbf V}^{\mathcal{C}}(\boldsymbol c)\,p_{\mathbf V,\boldsymbol c}(y)$. Since the Bell fusions and the boundary readout together form a complete projective measurement on the post-monitoring registers, the normalizer $\sum_{y'}w_{\mathbf V,\boldsymbol c}(y')$ is proportional to the Born probability of the accepted record $\boldsymbol c$, with a constant fixed by the branch normalization conventions; $p_{\mathbf V,\boldsymbol c}$ is therefore the exact Bayes conditional of the accepted transcript law. In the fixed-boundary convention of the numerics, $p_{\mathbf V,\boldsymbol c}$ is the separately normalized fixed-boundary branch law of the main text; the record-normalizer identity above is specific to the full-transcript convention and is not used in the hardness argument of Appendix~F. The classical sampler receives the description of the accepted pair $(\mathbf V,\boldsymbol c)$ and is asked to sample $p_{\mathbf V,\boldsymbol c}$. It need not reproduce a given $\boldsymbol c$.

In this formulation the average-case hardness assumption of Conjecture~1 is read over $(\mathbf V,\boldsymbol c)\sim\mathfrak A_{n,d,\kappa}$ and $y\sim\mathrm{Unif}(\mathcal Y)$, and the anti-concentration input of Theorem~2 becomes the pair-level statement that there exist constants $\gamma,\eta>0$, independent of $n$, with
\begin{equation}
\label{eq:F-ac}
\mathbb P_{\substack{(\mathbf V,\boldsymbol c)\sim\mathfrak A_{n,d,\kappa}\\ y\sim\mathrm{Unif}(\mathcal Y)}}\Bigl[p_{\mathbf V,\boldsymbol c}(y)\ge \tfrac{\gamma}{N}\Bigr]\ge\eta.
\end{equation}
Equation~\eqref{eq:F-ac} is the pair-level form of the main-text anti-concentration statement that a constant fraction of outcomes obey $p(\boldsymbol\beta\,|\,\boldsymbol c)=\Omega(d^{-2n})$ on typical accepted branches. These are the accepted-branch analogues of the average-case hardness and anti-concentration assumptions of boson sampling, with the record $\boldsymbol c$ in the role of the randomly selected sub-matrix and the hard object being the conditional branch probability rather than the record marginal.

\smallskip
\noindent\emph{Second-moment sufficient condition for anti-concentration.}
Fix either output convention defined above and write
$\Gamma(I):=\sum_{y\in\mathcal Y}p_I(y)^2$ for
$I=(\mathbf V,\boldsymbol c)\sim\mathfrak A_{n,d,\kappa}$.
If $Y\sim\mathrm{Unif}(\mathcal Y)$ and $X:=Np_I(Y)$, then
$\mathbb E[X]=1$ and
$\mathbb E[X^2]=\mathbb E_I[N\Gamma(I)]$.
Paley--Zygmund therefore gives, for every $0<\theta<1$,
\begin{equation}
\label{eq:F-PZ}
\mathbb P_{I,Y}\!\left[
p_I(Y)\ge\frac{\theta}{N}
\right]
\ge
\frac{(1-\theta)^2}
{\mathbb E_I[N\Gamma(I)]}.
\end{equation}
\noindent Consequently, if there exist constants $C<\infty$ and
$n_0$ such that
$\sup_{n\ge n_0}\mathbb E_I[N\Gamma(I)]\le C$,
then Eq.~\eqref{eq:F-ac} holds with
$\gamma=\theta$ and $\eta=(1-\theta)^2/C$.
Main-text Fig.~4 provides finite-size evidence for this sufficient condition
in the fixed-boundary convention, without establishing the asymptotic bound.

\paragraph*{A preassigned record is exponentially unlikely.}
Lemma~1 bounds only the total accepted-record probability: $\bar p_{\mathcal{C}}:=\mathbb E_{\mathbf V}\,p_{\mathcal{C}}(\mathbf V)=\Omega_{\kappa,d}(1)$ for $m=\lceil\kappa n^2\rceil$, and in fact $\bar p_{\mathcal{C}}\to\exp[-(1-1/d)/(2\kappa)]$. It does not assert that any particular record recurs with constant probability. By Haar invariance all records in $\mathcal{C}_{n,m}$ share the same average Born probability, so for any fixed $\boldsymbol c_0\in\mathcal{C}_{n,m}$,
\begin{equation}
\mathbb E_{\mathbf V}P_{\mathbf V}(\boldsymbol c_0)=\frac{\bar p_{\mathcal{C}}}{\binom{m}{n}}.
\end{equation}
With $m=\lceil\kappa n^2\rceil$ one has $\binom{m}{n}=\exp[n\log n+O(n)]$ and $\bar p_{\mathcal{C}}=\Theta(1)$, hence
\begin{equation}
\mathbb E_{\mathbf V}P_{\mathbf V}(\boldsymbol c_0)=\exp\bigl[-n\log n+O(n)\bigr].
\end{equation}
A preassigned record is therefore not the operational task: its probability is exponentially small on the $n\log n$ scale, far below the constant acceptance rate. The relevant object is the conditional law of the randomly realized record.

\paragraph*{From joint-transcript spoofing to conditional hardness.}
We next relate the conditional formulation to spoofing the full heralded transcript distribution: a good sampler for the full transcript induces good conditional branch laws on average.

\begin{lemma}[Joint-to-conditional total-variation transfer]
\label{lem:F-bridge}
Fix $\mathbf V$, and let $\Pi_{\mathbf V}(\boldsymbol c,y)=P_{\mathbf V}^{\mathcal{C}}(\boldsymbol c)\,p_{\mathbf V,\boldsymbol c}(y)$. For any joint law $Q_{\mathbf V}(\boldsymbol c,y)=Q_{\mathbf V,C}(\boldsymbol c)\,q_{\mathbf V,\boldsymbol c}(y)$, where $q_{\mathbf V,\boldsymbol c}$ is an arbitrary probability distribution on $\mathcal Y$ for records with $Q_{\mathbf V,C}(\boldsymbol c)=0$, if $d_{\mathrm{TV}}(\Pi_{\mathbf V},Q_{\mathbf V})\le\varepsilon_{\mathbf V}$, then
\begin{equation}
\mathbb E_{\boldsymbol c\sim P_{\mathbf V}^{\mathcal{C}}}\,d_{\mathrm{TV}}\bigl(p_{\mathbf V,\boldsymbol c},q_{\mathbf V,\boldsymbol c}\bigr)\le 2\varepsilon_{\mathbf V}.
\end{equation}
Consequently, if $\mathbb E_{\mathbf V}\,d_{\mathrm{TV}}(\Pi_{\mathbf V},Q_{\mathbf V})\le\varepsilon$, then $\mathbb E_{(\mathbf V,\boldsymbol c)\sim\mathfrak A_{n,d,\kappa}}\,d_{\mathrm{TV}}(p_{\mathbf V,\boldsymbol c},q_{\mathbf V,\boldsymbol c})\le 2\varepsilon$.
\end{lemma}

\begin{proof}
Fix $\mathbf V$ and abbreviate $P_C=P_{\mathbf V}^{\mathcal{C}}$, $p_{\boldsymbol c}=p_{\mathbf V,\boldsymbol c}$, $Q_C=Q_{\mathbf V,C}$, $q_{\boldsymbol c}=q_{\mathbf V,\boldsymbol c}$. Since each $q_{\boldsymbol c}$ is a probability distribution, $\|q_{\boldsymbol c}\|_1=1$, and for each $\boldsymbol c$,
\begin{equation}
P_C(\boldsymbol c)\|p_{\boldsymbol c}-q_{\boldsymbol c}\|_1\le\|P_C(\boldsymbol c)p_{\boldsymbol c}-Q_C(\boldsymbol c)q_{\boldsymbol c}\|_1+|P_C(\boldsymbol c)-Q_C(\boldsymbol c)|.
\end{equation}
Summing over $\boldsymbol c$ and using that taking marginals does not increase the $\ell_1$ distance,
\begin{equation}
\sum_{\boldsymbol c}P_C(\boldsymbol c)\|p_{\boldsymbol c}-q_{\boldsymbol c}\|_1\le\|\Pi_{\mathbf V}-Q_{\mathbf V}\|_1+\|P_C-Q_C\|_1\le 2\|\Pi_{\mathbf V}-Q_{\mathbf V}\|_1.
\end{equation}
Since $\mathbb E_{\boldsymbol c\sim P_C}\,d_{\mathrm{TV}}(p_{\boldsymbol c},q_{\boldsymbol c})=\tfrac12\sum_{\boldsymbol c}P_C(\boldsymbol c)\|p_{\boldsymbol c}-q_{\boldsymbol c}\|_1$ and $\|\Pi_{\mathbf V}-Q_{\mathbf V}\|_1=2\,d_{\mathrm{TV}}(\Pi_{\mathbf V},Q_{\mathbf V})$, the first bound follows; averaging over $\mathbf V$ gives the second.
\end{proof}

If the oracle outputs joint transcripts, Stockmeyer's approximate counting estimates both $Q_{\mathbf V}(\boldsymbol c,y)$ and its marginal $Q_{\mathbf V,C}(\boldsymbol c)$, and their ratio estimates the conditional $q_{\mathbf V,\boldsymbol c}(y)$ wherever the marginal is nonzero, with the contribution of records of vanishing spoofer marginal controlled by Lemma~\ref{lem:F-bridge}. Together with the conditional-sampling hardness proved in Appendix~F, this shows that a good joint-transcript spoofer also yields the conditional branch estimates ruled out in Appendix~F, so the hardness conclusion of Theorem~2 extends from conditional samplers to spoofers of the full heralded transcript law $\Pi_{\mathbf V}$, with a constant-factor loss in the error budget.

\paragraph*{Numerical check.}
To confirm that the anti-concentration diagnostics of the main text are not an artifact of treating the record as uniform, we compare the operational Born-sampled record against a uniform-record control on the same instance. For each Haar draw $\mathbf V$ we form $\boldsymbol c_{\rm Born}\sim P_{\mathbf V}^{\mathcal{C}}$ and $\boldsymbol c_{\rm unif}\sim\mathrm{Unif}(\mathcal{C}_{n,m})$, evaluate the conditional collision moment
\begin{equation}
\Gamma(\mathbf V,\boldsymbol c)=\sum_{y\in\mathcal Y}p_{\mathbf V,\boldsymbol c}(y)^2
\end{equation}
for each record, and report the paired difference
\begin{equation}
\delta_\Gamma(\mathbf V):=\frac{\Gamma(\mathbf V,\boldsymbol c_{\rm Born})-\Gamma(\mathbf V,\boldsymbol c_{\rm unif})}{\Gamma_{\rm Haar}},\qquad \Gamma_{\rm Haar}=\frac{2}{d^{2n}+1}
\end{equation}
in the fixed-$r$ convention. Pairing on a common $\mathbf V$ cancels most instance-to-instance fluctuation. The record here is the coarse block-occupation record. Internal labels are summed over in forming the Born marginal and are not additional measured data. For each $n,d$, define the sample mean
\begin{equation}
\Delta_\Gamma:=\overline{\delta_\Gamma}
=\frac{1}{N_{\rm inst}}\sum_{i=1}^{N_{\rm inst}}\delta_\Gamma(\mathbf V_i),
\end{equation}
with $N_{\rm inst}=200$ instances per size, and let $\mathrm{SEM}(\delta_\Gamma)$ denote its paired standard error. Fig.~\ref{fig:D3a} plots the standardized mean $z:=\Delta_\Gamma/\mathrm{SEM}(\delta_\Gamma)$; where the record is forced ($n=2$, $m=n$) the two records coincide and $\delta_\Gamma$ vanishes identically, and at the remaining accessible sizes $z$ is compatible with zero. Thus, over the accessible sizes, Born sampling of the accepted record does not visibly alter the collision-moment evidence for anti-concentration reported in main-text Fig.~4.

\begin{figure}[t]
\centering
\includegraphics[width=0.58\textwidth]{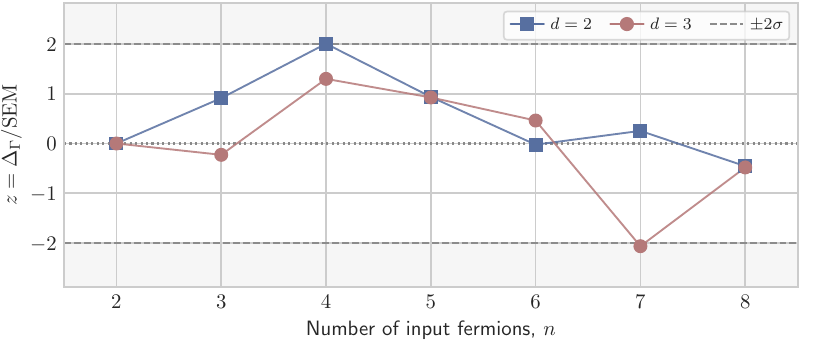}
\caption{Born sampling of the accepted record does not visibly change the anti-concentration diagnostic. For each Haar instance $\mathbf V$, a Born-sampled record $\boldsymbol c_{\rm Born}\sim P_{\mathbf V}^{\mathcal{C}}$ and a uniform-record control $\boldsymbol c_{\rm unif}\sim\mathrm{Unif}(\mathcal{C}_{n,m})$ are formed on the same $\mathbf V$, and their conditional collision moments, relative to the Haar value $\Gamma_{\rm Haar}=2/(d^{2n}+1)$, define the per-instance paired difference $\delta_\Gamma(\mathbf V)=[\Gamma(\mathbf V,\boldsymbol c_{\rm Born})-\Gamma(\mathbf V,\boldsymbol c_{\rm unif})]/\Gamma_{\rm Haar}$. To compare sizes uniformly we plot the standardized mean $z=\Delta_\Gamma/\mathrm{SEM}(\delta_\Gamma)$, where $\Delta_\Gamma:=\overline{\delta_\Gamma}$ is the sample mean of $\delta_\Gamma$ over the $200$ instances at each size and $\mathrm{SEM}(\delta_\Gamma)$ its paired standard error. Dotted line $z=0$, dashed lines $z=\pm2$. Data are shown for $d=2$ (blue squares) and $d=3$ (red circles) at $\kappa=1/2$. Where the record is forced ($n=2$, $m=n$) the two records coincide and $\delta_\Gamma$ vanishes identically. All standardized means are compatible with zero within finite-size uncertainty (per-series $\chi^2$ over the nontrivial sizes gives $\chi^2/\mathrm{dof}=1.00$, $p=0.42$ for $d=2$ and $\chi^2/\mathrm{dof}=1.22$, $p=0.29$ for $d=3$; the largest excursions in either series, $d=2$ at $n=4$ and $d=3$ at $n=7$, lie at about $2\sigma$), so for the sizes studied Born sampling of the record does not visibly alter the collision-moment diagnostic relative to the uniform-record control.}
\label{fig:D3a}
\end{figure}

\section*{Appendix F: Proof of Theorem~2 (Conditional hardness of approximate sampling)}

We prove Theorem~2 in the accepted-branch formulation of Appendix~E. Throughout, $(\mathbf V,\boldsymbol c)\sim\mathfrak A_{n,d,\kappa}$, the classical sampler is given the accepted pair $(\mathbf V,\boldsymbol c)$, and the target is the conditional branch law $p_{\mathbf V,\boldsymbol c}$ on $\mathcal Y$. The argument below uses only $N:=|\mathcal Y|$ and the constants in Eq.~\eqref{eq:F-ac}, so it applies verbatim in either convention of Appendix~E, with Conjecture~1 and the anti-concentration constants read in that same convention; the two conventions are formally distinct assumptions. In the fixed-boundary convention of the numerics $N=d^{2n}$, and including the boundary readout multiplies $N$ by $d$.

\begin{proof}[Proof of Theorem~2]
Assume, towards a contradiction, that there exists a classical probabilistic polynomial-time sampler which, on input $(\mathbf V,\boldsymbol c)$, samples from $q_{\mathbf V,\boldsymbol c}$ satisfying
\begin{equation}
\mathbb E_{(\mathbf V,\boldsymbol c)\sim\mathfrak A_{n,d,\kappa}}
d_{\mathrm{TV}}\!\left(p_{\mathbf V,\boldsymbol c},q_{\mathbf V,\boldsymbol c}\right)\le\varepsilon,
\end{equation}
with $\varepsilon\le \eta\gamma/(4\,\poly(n))$, where $\gamma,\eta>0$ are the constants in the pair-level anti-concentration condition~\eqref{eq:F-ac}.

By Stockmeyer's approximate-counting theorem, oracle access to this sampler gives a $\mathrm{BPP}^{\mathrm{NP}}$ procedure that multiplicatively estimates $q_{\mathbf V,\boldsymbol c}(y)$ for any specified accepted pair $(\mathbf V,\boldsymbol c)$ and output $y\in\mathcal Y$. Let $e(\mathbf V,\boldsymbol c,y):=|p_{\mathbf V,\boldsymbol c}(y)-q_{\mathbf V,\boldsymbol c}(y)|$. Since $\sum_y e(\mathbf V,\boldsymbol c,y)=2\,d_{\mathrm{TV}}(p_{\mathbf V,\boldsymbol c},q_{\mathbf V,\boldsymbol c})$, the average-TV assumption implies
\begin{equation}
\mathbb E_{\substack{(\mathbf V,\boldsymbol c)\sim\mathfrak A_{n,d,\kappa}\\ y\sim\mathrm{Unif}(\mathcal Y)}}e(\mathbf V,\boldsymbol c,y)\le\frac{2\varepsilon}{N}.
\end{equation}
By Markov's inequality, with probability at least $1-\eta/2$ over the pair $((\mathbf V,\boldsymbol c),y)$,
\begin{equation}
e(\mathbf V,\boldsymbol c,y)\le\frac{4\varepsilon}{\eta N}.
\end{equation}
By the pair-level anti-concentration condition~\eqref{eq:F-ac}, with probability at least $\eta$,
\begin{equation}
p_{\mathbf V,\boldsymbol c}(y)\ge\frac{\gamma}{N}.
\end{equation}
Hence both events hold simultaneously on a set of accepted instance--output pairs of probability at least $\eta/2$, and there
\begin{equation}
\frac{|p_{\mathbf V,\boldsymbol c}(y)-q_{\mathbf V,\boldsymbol c}(y)|}{p_{\mathbf V,\boldsymbol c}(y)}\le\frac{4\varepsilon}{\eta\gamma}\leq\frac{1}{\poly(n)} .
\end{equation}
Thus, on a constant fraction of accepted instance--output pairs, $q_{\mathbf V,\boldsymbol c}(y)$ is a relative-error approximation to $p_{\mathbf V,\boldsymbol c}(y)$. Combining this with the Stockmeyer estimate of $q_{\mathbf V,\boldsymbol c}(y)$ yields a $\mathrm{BPP}^{\mathrm{NP}}$ procedure estimating $p_{\mathbf V,\boldsymbol c}(y)$ within relative error $1/\poly(n)$ on a constant fraction of accepted instance--output pairs, contradicting Conjecture~1 unless the polynomial hierarchy collapses.

If a classical algorithm instead spoofs the full heralded transcript law $\Pi_{\mathbf V}$, Lemma~\ref{lem:F-bridge} converts an average joint total-variation bound into an average conditional one, with only a factor-two loss, and Stockmeyer's approximate counting also estimates the joint output probability and the record marginal of the spoofer, hence their conditional ratio, as explained after Lemma~\ref{lem:F-bridge}. The same contradiction therefore applies to full-transcript spoofers within the corresponding error budget.
\end{proof}

\section*{Appendix G: Worst-case exact-value hardness benchmarks}
\label{app:worstcase-hardness}

This appendix collects two complementary \emph{worst-case exact-value} benchmarks relevant to the algebraic complexity landscape surrounding monitored-FLO branches. These benchmarks provide exact algebraic anchors for the branch-polynomial structure, but they do not establish the average-case hardness or the anti-concentration used for the sampling claim in the main text, from which they are formally distinct.

First, we recall a canonical family of non-commutative polynomials, namely Cayley’s row-ordered determinant, whose exact evaluation is $\#\mathrm P$-hard in the worst case even for constant-size matrix entries.
Second, we give a protocol-adapted benchmark based on cyclically closing the boundary wire:
for a specially restricted subfamily, this cyclic-closure value reduces to the fermionant, which is $\#\mathrm P$-hard for
fixed $d\ge3$ (and $\oplus\mathrm P$-hard for $d=2$). On the same restricted slice the open-boundary branch operator is itself a scalar multiple of the identity, so the cyclic closure coincides with the physical open-boundary branch amplitude $\mathcal A_d$ for fixed equal boundary states, and the fermionant reduction therefore also gives a worst-case exact-value benchmark for $\mathcal A_d$ on that slice. For generic matrix-valued branches, $\mathcal Z_d$ and $\mathcal A_d$ are different observables. The reductions below are exact-value, worst-case statements rather than average-case sampling claims.

\paragraph*{Remark.}
The commuting-control data in main-text Figs.~3--4 are included solely to isolate the role of measurement-induced non-commutativity for typical open-boundary branches via MPO-bond and anti-concentration diagnostics. The worst-case scalar-slice benchmark below serves a different purpose: it shows that the measurement-enforced ordering alone already realizes a hard ordered cycle sum in a restricted exact-value setting, with matrix non-commutativity set to zero but the measurement-enforced cycle weighting retained. The two are distinct objects, a typical-random diagnostic versus a worst-case exact-value reduction, and the latter does not enter the main-text sampling-hardness discussion and should not be read as a typical-branch statement.

\medskip
Let $\mathcal R$ be an associative (possibly non-commutative) algebra and let $\mathbf B=(B_{i,j})\in \M_N(\mathcal R)$.
Cayley’s row-ordered non-commutative determinant is the non-commutative polynomial
\begin{equation}
\label{eq:cayley-det-appE}
\det_{\rm Cay}(\mathbf B)\ :=\ \sum_{\sigma\in \mathfrak S_N}\sgn(\sigma)
B_{1,\sigma(1)} B_{2,\sigma(2)}\cdots B_{N,\sigma(N)} \ \in\ \mathcal R .
\end{equation}

\setcounter{proposition}{1}
\begin{proposition}[Known worst-case hardness of Cayley’s determinant over $2\times 2$ matrix entries~\cite{Chien2011Almost,Arvind2010On,Gentry2014Noncommutative}]
\label{prop:ncdet-hard-appE}
Let $\mathbb F$ be a field of characteristic $p\ge 0$.
Given an $N\times N$ matrix whose entries lie in $\M_2(\mathbb F)$, exactly computing $\det_{\rm Cay}$ is
$\#\mathrm P$-hard under polynomial-time Turing reductions when $p=0$, and $\mathrm{Mod}_p\mathrm P$-hard
under polynomial-time Turing reductions when $p>0$ is odd.
\end{proposition}

Proposition~\ref{prop:ncdet-hard-appE} is included as structural rather than reduction-theoretic context.
In monitored-FLO, each permutation contribution to the branch operator
$\mathcal{T}_{\boldsymbol\beta}(\mathbf S)$ consists of an ordered product of dressed blocks along the unique open boundary path, together with additional scalar trace-loop factors from closed wiring cycles. See Proposition~1.
Thus $\mathcal{T}_{\boldsymbol\beta}(\mathbf S)$ belongs to the same broader landscape of ordered non-commutative polynomial expressions, albeit enriched here by trace-loop structure.

\medskip
Fix a collision-free record $\boldsymbol c$ and fusion outcomes $\boldsymbol\beta$, and let
$\mathcal{T}_{\boldsymbol\beta}(\mathbf S)\in \M_d(\C)$ denote the corresponding open-boundary branch operator from the main text, defined through
\begin{equation}
\mathcal A_d(\boldsymbol c,{\boldsymbol\beta})=\langle r|\mathcal{T}_{\boldsymbol\beta}(\mathbf S)|\ell\rangle,
\qquad
\widetilde S_{t,k}:=S_{t,k}^{\top}\hat D_{\beta_t}.
\end{equation}
To formulate the second worst-case benchmark, we now introduce a different scalar observable obtained by cyclically closing the boundary wire of $\mathcal{T}_{\boldsymbol\beta}(\mathbf S)$.

Introduce a reference qudit $R\simeq\C^d$ and prepare $|\Phi_d^+\rangle_{R,A_0}$ in place of a pure boundary state on $A_0$.
After running the same monitored-FLO evolution and the same Bell-fusion layer producing the branch operator
$\mathcal{T}_{\boldsymbol\beta}(\mathbf S)$ on the boundary leg $A_n$, close the boundary by projecting $(R,A_n)$ onto
$\langle\Phi_d^+|$.
Define the resulting (unnormalized) cyclic-closure scalar
\begin{equation}
\label{eq:Z-cyclic-def-appE}
\mathcal Z_d(\boldsymbol c,{\boldsymbol\beta})
:=
\langle\Phi_d^+|_{R,A_n}\Bigl(\id_R\otimes \mathcal{T}_{\boldsymbol\beta}(\mathbf S)\Bigr)|\Phi_d^+\rangle_{R,A_0}.
\end{equation}
For generic matrix-valued branches, $\mathcal Z_d(\boldsymbol c,{\boldsymbol\beta})$ is distinct from the open-boundary branch amplitude $\mathcal A_d(\boldsymbol c,{\boldsymbol\beta})$ governing the sampling problem in the main text, and is introduced here as a worst-case exact-value benchmark. On the scalar slice used below, however, the branch operator becomes a scalar multiple of the identity, and then $\mathcal Z_d$ coincides with $\mathcal A_d$ for fixed equal boundary states.
For any $\Xi\in\M_d(\C)$,
\begin{equation}
\label{eq:Phi-trace-id-appE}
\langle\Phi_d^+|(\id\otimes \Xi)|\Phi_d^+\rangle = \frac{1}{d}\Tr(\Xi),
\end{equation}
and hence, the cyclic-closure scalar can be written 
as
\begin{equation}
\label{eq:Z-is-trace-appE}
\mathcal Z_d(\boldsymbol c,{\boldsymbol\beta})
=
\frac{1}{d}\Tr\big(\mathcal{T}_{\boldsymbol\beta}(\mathbf S)\big),
\end{equation}
where $\Tr$ is the usual trace on $\M_d(\C)$.

\smallskip
In the path--cycle expansion of Proposition~1, each FLO permutation term
$\sigma\in\mathfrak S_n$ induces directed edges
\begin{equation}
\label{eq:edges-open-appE}
A_{t-1}\xrightarrow{\ \widetilde S_{t,\sigma(t)}\ } A_{\sigma(t)},
\qquad t=1,\dots,n.
\end{equation}
After cyclic closure, the previously open $A_0\to A_n$ path is also traced and becomes an additional cycle.

Let $\#\mathrm{cyc}(\pi)$ denote the number of cycles of a permutation $\pi\in\mathfrak S_n$, and define the fixed $n$-cycle
\begin{equation}
\label{eq:pi0-def-appE}
\pi_0\in\mathfrak S_n,\qquad
\pi_0(t):=
\begin{cases}
t-1,& t=2,\dots,n,\\
n,& t=1,
\end{cases}
\qquad\text{equivalently }\ \pi_0=(1\ n\ n{-}1\ \cdots\ 2),
\end{equation}
so $\sgn(\pi_0)=(-1)^{n-1}$.

After cyclic closure, the wiring sends each vertex $A_j$ (for $j=1,\dots,n$) to
$A_{\sigma(\pi_0^{-1}(j))}$, where $\pi_0^{-1}$ is the forward cyclic shift on
$\{1,\dots,n\}$. Hence the directed cycles in the cyclic-closure wiring are exactly
the cycles of $\sigma\circ\pi_0^{-1}$, so the number of directed cycles after closure is
$\#\mathrm{cyc}(\sigma\circ\pi_0^{-1})$.

\medskip
We now exhibit a restricted subfamily of instances for which $\mathcal Z_d$ reduces to the fermionant.

\begin{definition}[Fermionant~\cite{Mertens2013The}]
\label{def:fermionant-appE}
For $k\in\C$ and $\mathbf W\in \M_n(\C)$, the fermionant is
\begin{equation}
\label{eq:fermionant-def-appE}
\mathrm{Ferm}_k(\mathbf W)
:=(-1)^n\sum_{\pi\in\mathfrak S_n}(-k)^{\#\mathrm{cyc}(\pi)}\prod_{t=1}^n W_{t,\pi(t)}
=\sum_{\pi\in\mathfrak S_n}\sgn(\pi)k^{\#\mathrm{cyc}(\pi)}\prod_{t=1}^n W_{t,\pi(t)}.
\end{equation}
\end{definition}

Fix ${\boldsymbol\beta}=\mathbf 0$, so that $\hat D_{\beta_t}=\id_d$ for all $t$. Since this section concerns worst-case exact-value evaluation, the probability of obtaining this particular fusion outcome is irrelevant.
Assume the dressed blocks are scalar multiples of the identity,
\begin{equation}
\label{eq:scalar-blocks-appE}
\widetilde S_{t,k}=a_{t,k}\id_d
\qquad (1\le t,k\le n),
\end{equation}
for some scalars $a_{t,k}\in\C$ (e.g., internal-state-independent propagators $\mathbf V=U\otimes \id_d$).

Define $\mathbf W\in\M_n(\C)$ by the fixed row shift
\begin{equation}
\label{eq:W-rowshift-appE}
W_{t,k}:=a_{\pi_0^{-1}(t),k},
\end{equation}
with $\pi_0$ from~\eqref{eq:pi0-def-appE}.
Under~\eqref{eq:scalar-blocks-appE}, the ordered label product around any directed cycle $C$ equals
\(\Pi_C=\Bigl(\prod_{t\in C} a_{t,\sigma(t)}\Bigr)\id_d\), and hence contributes the trace factor $\Tr(\Pi_C)=d\prod_{t\in C} a_{t,\sigma(t)}$.
Using~\eqref{eq:Z-is-trace-appE} together with the above identification of the directed cycles after closure with the cycles of $\sigma\circ\pi_0^{-1}$, the cyclic-closure value admits the explicit permutation expansion
\begin{equation}
\label{eq:Z-sigma-sum-appE}
\mathcal Z_d(\boldsymbol c,\mathbf 0)
=
\frac{1}{d^{n+1}}\sum_{\sigma\in\mathfrak S_n}\sgn(\sigma)\,d^{\#\mathrm{cyc}(\sigma\circ\pi_0^{-1})}\prod_{t=1}^n a_{t,\sigma(t)},
\end{equation}
where $d^{-n}$ is the fixed Bell-projection prefactor already absorbed into $\mathcal{T}_{\mathbf 0}$ and the extra $d^{-1}$ is
from~\eqref{eq:Z-is-trace-appE}.
Making the bijective change of variables $\pi=\sigma\circ\pi_0^{-1}$ (hence $\sigma=\pi\circ\pi_0$) and reindexing rows gives
\begin{equation}
\prod_{t=1}^n a_{t,\sigma(t)}
=\prod_{t=1}^n a_{t,\pi(\pi_0(t))}
=\prod_{s=1}^n a_{\pi_0^{-1}(s),\pi(s)}
=\prod_{s=1}^n W_{s,\pi(s)}.
\end{equation}
Moreover, $\sgn(\sigma)=\sgn(\pi)\sgn(\pi_0)$ and $\#\mathrm{cyc}(\sigma\circ\pi_0^{-1})=\#\mathrm{cyc}(\pi)$, so the permutation sum in~\eqref{eq:Z-sigma-sum-appE} becomes a fermionant sum.
Thus, the compact fermionant form
\begin{equation}
\label{eq:Z-equals-fermionant-appE}
\mathcal Z_d(\boldsymbol c,\mathbf 0)
= \frac{\sgn(\pi_0)\,\mathrm{Ferm}_d(\mathbf W)}{d^{n+1}}
\end{equation}
holds.

The same scalar slice fixes the physical open-boundary amplitude. On the slice~\eqref{eq:scalar-blocks-appE} every dressed block is a scalar multiple of $\id_d$, so every ordered path product and every loop factor is a scalar, and no matrix structure distinguishes the permutation terms. The branch operator is therefore itself scalar, $\mathcal{T}_{\mathbf 0}(\mathbf S)=C(a)\,\id_d$. For the fixed equal-boundary convention $|\ell\rangle=|r\rangle=|0\rangle$ used in the numerics and in this reduction, this gives
\begin{equation}
\label{eq:A-equals-fermionant-appE}
\mathcal A_d(\boldsymbol c,\mathbf 0)
=\langle 0|\mathcal{T}_{\mathbf 0}(\mathbf S)|0\rangle
=C(a)
=\tfrac{1}{d}\Tr\bigl(\mathcal{T}_{\mathbf 0}(\mathbf S)\bigr)
=\mathcal Z_d(\boldsymbol c,\mathbf 0)
=\frac{\sgn(\pi_0)\,\mathrm{Ferm}_d(\mathbf W)}{d^{n+1}} ,
\end{equation}
so the fermionant identity~\eqref{eq:Z-equals-fermionant-appE} holds verbatim for the physical open-boundary amplitude on this slice; for unequal boundary states $\mathcal A_d=C(a)\langle r|\ell\rangle$, which is why the clean identification is stated for fixed equal boundary states. No cyclic closure is required on this slice.

We recall the known worst-case hardness of the fermionant.

\setcounter{proposition}{2}
\begin{proposition}[Known worst-case hardness of the fermionant~\cite{Mertens2013The}]
\label{prop:fermionant-hard-appE}
Let $k\ge 2$ be a fixed constant.
Exact evaluation of $\mathrm{Ferm}_k(\mathbf W)$ is $\#\mathrm P$-hard under polynomial-time Turing reductions for any $k>2$,
and is $\oplus\mathrm P$-hard under polynomial-time Turing reductions for $k=2$.
Moreover, both statements hold even when $\mathbf W$ is restricted to be the adjacency matrix of a planar graph.
\end{proposition}

Combining Proposition~\ref{prop:fermionant-hard-appE} with Eqs.~\eqref{eq:Z-equals-fermionant-appE} and~\eqref{eq:A-equals-fermionant-appE} immediately gives the following corollary.

\begin{corollary}[Worst-case hardness of cyclic-closure and scalar-slice open-boundary values]
\label{cor:Z-hard-appE}
For any fixed $d\ge 3$, exact computation of the cyclic-closure scalar $\mathcal Z_d(\boldsymbol c,\mathbf 0)$ is $\#\mathrm P$-hard under polynomial-time Turing reductions, already on the restricted subfamily~\eqref{eq:scalar-blocks-appE}. On the same slice and for fixed equal boundary states, the physical open-boundary branch amplitude $\mathcal A_d(\boldsymbol c,\mathbf 0)$ coincides with $\mathcal Z_d(\boldsymbol c,\mathbf 0)$ by~\eqref{eq:A-equals-fermionant-appE}, so its exact computation is also $\#\mathrm P$-hard for fixed $d\ge 3$. For $d=2$ the corresponding exact computation is $\oplus\mathrm P$-hard.
\end{corollary}

\smallskip
\noindent\emph{Remark (unitarity constraint).}
Corollary~\ref{cor:Z-hard-appE} is stated directly in terms of the selected scalar submatrix $(a_{t,k})$.
If one additionally wishes to realize such instances from a valid FLO propagator, note that the fermionant is
homogeneous of degree $n$: $\mathrm{Ferm}_d(\lambda \mathbf W)=\lambda^n\mathrm{Ferm}_d(\mathbf W)$.
Thus we may rescale $\mathbf W\mapsto \lambda \mathbf W$ so that $\|\lambda \mathbf W\|\le 1$ and embed $\lambda \mathbf W$ as a principal block of a unitary via
a standard unitary completion (e.g., the Halmos unitary dilation of a contraction).
By fixed input/output mode permutations (unitary), one can place the desired $n\times n$ block at the locations selected by
the collision-free record $\boldsymbol c$.
The resulting unitary completion may involve algebraic-number entries; it is used only as an existence statement of a realizing FLO propagator, while the input of Corollary~1 remains the selected block data, so no explicit polynomial-size description of the completion is required or claimed. The same completion realizes the scalar-slice open-boundary amplitude~\eqref{eq:A-equals-fermionant-appE} for fixed equal boundary states.

\end{document}